\documentclass[11pt]{article}
\usepackage{graphics,latexsym}
 \usepackage{citesort}
\newcounter{savenumi}

\newtheorem{theoremfoo}{Theorem}%[section] %by chapter in report style
\newenvironment{theorem}{\pagebreak[1]\begin{theoremfoo}}{\end{theoremfoo}}

\newtheorem{propositionfoo}[theoremfoo]{Proposition}

\newtheorem{lemmafoo}[theoremfoo]{Lemma}
\newenvironment{lemma}{\pagebreak[1]\begin{lemmafoo}}{\end{lemmafoo}}
\newtheorem{conjecturefoo}[theoremfoo]{Conjecture}

\newtheorem{corollaryfoo}[theoremfoo]{Corollary}
\newenvironment{corollary}{\pagebreak[1]\begin{corollaryfoo}}{\end{corollaryfoo}}
\newtheorem{exercisefoo}{Exercise}

\newtheorem{openfoo}[theoremfoo]{Question}

\newtheorem{nttn}[theoremfoo]{Notation}

\newtheorem{dfntn}[theoremfoo]{Definition}
\newenvironment{definition}{\pagebreak[1]\begin{dfntn}\rm}{\end{dfntn}}

\newenvironment{proof}
    {\pagebreak[1]{\narrower\noindent {\bf Proof:\quad\nopagebreak}}}{\QED}
\newcommand{\floor}[1]{\left\lfloor#1\right\rfloor}

\def\nre.{$n$\/-r.e.}
\newcommand{\scrod}{\quad\nopagebreak}

\newtheorem{factfoo}[theoremfoo]{Fact}

\newtheorem{propertyfoo}[theoremfoo]{Property}

\makeatletter

\def\@makechapterhead#1{ \vspace*{50pt} { \parindent 0pt \raggedright 
 \ifnum \c@secnumdepth >\m@ne \huge\bf \@chapapp{} \thechapter. \par 
 \vskip 20pt \fi \Huge \bf #1\par 
 \nobreak \vskip 40pt } }

\def\@sect#1#2#3#4#5#6[#7]#8{\ifnum #2>\c@secnumdepth
     \def\@svsec{}\else 
     \refstepcounter{#1}\edef\@svsec{\csname the#1\endcsname.\hskip 1em }\fi
     \@tempskipa #5\relax
      \ifdim \@tempskipa>\z@ 
        \begingroup #6\relax
          \@hangfrom{\hskip #3\relax\@svsec}{\interlinepenalty \@M #8\par}
        \endgroup
       \csname #1mark\endcsname{#7}\addcontentsline
         {toc}{#1}{\ifnum #2>\c@secnumdepth \else
                      \protect\numberline{\csname the#1\endcsname}\fi
                    #7}\else
        \def\@svsechd{#6\hskip #3\@svsec #8\csname #1mark\endcsname
                      {#7}\addcontentsline
                           {toc}{#1}{\ifnum #2>\c@secnumdepth \else
                             \protect\numberline{\csname the#1\endcsname}\fi
                       #7}}\fi
     \@xsect{#5}}

\def\@begintheorem#1#2{\it \trivlist \item[\hskip \labelsep{\bf #1\ #2.}]}

\def\@opargbegintheorem#1#2#3{\it \trivlist
      \item[\hskip \labelsep{\bf #1\ #2\ (#3).}]}

\makeatother

\newif\ifshortconferences
\shortconferencesfalse
\newif\ifmediumconferences
\mediumconferencesfalse

\def\ending#1{{\count1=#1\relax
\count2=\count1
\divide\count2 by 100
\multiply\count2 by 100
\advance\count1 by -\count2
\ifnum\count1=11
th%
\else \ifnum\count1=12
th%
\else \ifnum\count1=13
th%
\else 
\count2=\count1
\divide\count1 by 10
\multiply\count1 by 10
\advance\count2 by -\count1
\ifnum\count2=1
st%
\else \ifnum\count2=2
nd%
\else \ifnum\count2=3
rd%
\else th%
\fi\fi\fi\fi\fi\fi
}}

\def\Proceedingsofthe{\ifshortconferences Proc.\else\ifmediumconferences Proc.\else Proceedings of the\fi\fi}

\newcounter{confnum}

\def\conf#1#2{%
\setcounter{confnum}{#2}%
\addtocounter{confnum}{-\csname #1zero\endcsname}%
\ifnum\value{confnum}=1%
\expandafter\ifx\csname #1One\endcsname\relax%
\Proceedingsofthe\ \arabic{confnum}\ending{\value{confnum}}\ \csname #1name\endcsname%
\else \csname #1One\endcsname\fi%
\else%
\Proceedingsofthe\
\arabic{confnum}\ending{\value{confnum}}\ \csname #1name\endcsname\fi}

\def\qsym{\vrule width0.7ex height0.9em depth0ex}

\newif\ifqed\qedtrue

\def\noqed{\global\qedfalse}

\def\qed{\ifqed{\penalty1000\unskip\nobreak\hfil\penalty50
\hskip2em\hbox{}\nobreak\hfil\qsym
\parfillskip=0pt \finalhyphendemerits=0\par\medskip}\fi\global\qedtrue}

\makeatletter
\def\eqnqed{\noqed
	\def\@tempa{equation}
	\ifx\@tempa\@currenvir\def\@eqnnum{\qsym}%
	\addtocounter{equation}{-1}\else%
    \def\@@eqncr{\let\@tempa\relax
    \ifcase\@eqcnt \def\@tempa{& & &}\or \def\@tempa{& &}%
      \else \def\@tempa{&}\fi
     \@tempa {\def\@eqnnum{{\qsym}}\@eqnnum}%  [changed theequation to @eqnnum]
     \global\@eqnswtrue\global\@eqcnt\z@\cr}\fi}

\def\eqnlabel#1#2{\if@filesw {\let\thepage\relax%
   \def\protect{\noexpand\noexpand\noexpand}%
   \edef\@tempa{\write\@auxout{\string
      \newlabel{#2}{{{#1}}{\thepage}}}}%
   \expandafter}\@tempa%
   \if@nobreak \ifvmode\nobreak\fi\fi\fi%
	\def\@tempa{equation}
	\ifx\@tempa\@currenvir\def\theequation{{#1}}% 
	\addtocounter{equation}{-1}\else%
    \def\@@eqncr{\let\@tempa\relax
    \ifcase\@eqcnt \def\@tempa{& & &}\or \def\@tempa{& &}%
      \else \def\@tempa{&}\fi
     \@tempa {\def\@eqnnum{{#1}}\@eqnnum}% [changed theequation to @eqnnum]
     \global\@eqnswtrue\global\@eqcnt\z@\cr}\fi}

\makeatother

\def\QED{\qed}

\makeatother

\usepackage[top=1.in, bottom=0.9in, left=1.0in, right=1.0in]{geometry}
\newcommand{\shift}{{\rm shift}}

\newcommand{\diff}{{\rm diff}}
\newcommand{\match}{{\rm Match}}

\newcommand{\algmnam}{Recover-Motif}
\newcommand{\algmname}{Algorithm~\algmnam}
\newcommand{\algma}{\algmname~}
\newcommand{\roughleft}{{\rm roughLeft}}
\newcommand{\roughright}{{\rm roughRight}}
\newcommand{\phaseone}{Boundary-Phase}
\newcommand{\phasetwo}{Extract-Phase}
\newcommand{\phasethree}{Voting-Phase}

\newcommand{\LB}{{\rm LB}}
\newcommand{\RB}{{\rm RB}}
\newcommand{\thresholdL}{{(\log n)^{3+\tau}\over 100}}
\newcommand{\algtype}{{\rm algorithm-type}}
\newcommand{\sublinear}{{\rm RANDOMIZED-SUBLINEAR}}
\newcommand{\randomized}{{\rm RANDOMIZED-SUBQUADRATIC}}
\newcommand{\deterministic}{{\rm DETERMINISTIC-SUPERQUADRATIC}}

\newcommand{\prob}{{\rm Pr}}

\begin{document}
\setcounter{page}{1}
\date{March 7, 2012}

\title{Sublinear Time Motif Discovery from Multiple Sequences\thanks{This research is supported in part by
National Science Foundation Early Career Award 0845376.}}

\author{ Bin Fu, Yunhui Fu,  and Yuan Xue\\
Department of Computer Science, University of Texas--Pan American\\
Edinburg, TX 78541, USA\\
Emails: binfu@cs.panam.edu,  fuyunhui@gmail.com, xuey@utpa.edu}
\maketitle

\begin{abstract}
A natural  probabilistic model for motif discovery has been used to
experimentally test the quality of motif discovery
programs. %~\cite{PevznerSze00,KeichPevzner02,KeichPevzner02b,WangDong05,ChinLeung05}.
In this model, there are $k$ background sequences, and each
character in a background sequence is a random character from an
alphabet $\Sigma$. A motif $G=g_1g_2\ldots g_m$ is a string of $m$
characters. Each background sequence is implanted a
probabilistically generated approximate copy of  $G$. For a
probabilistically generated approximate copy $b_1b_2\ldots b_m$ of
$G$, every character $b_i$ is probabilistically  generated such that
the probability for $b_i\neq g_i$ is at most $\alpha$. We develop
three algorithms that under the probabilistic model can find the
implanted motif with high probability via a tradeoff between
computational time and the probability of mutation. Each algorithm
has the preprocessing part and the voting part. We use a pair of
function $(t_1(n,k), t_2(n,k))$ to describe the computational
complexity of motif detection algorithm, where $n$ is the largest
length of input sequence, and $k$ is the number of sequences.
Function $t_1(n,k)$ is the time complexity for the part for
preprocessing  and $t_2(n,k)$ is the time complexity for recovering
one character for motif after preprocessing. The total time is
$O(t_1(n,k)+t_2(n,k)|G|)$.

(1) There exists a randomized algorithm such that there are positive
constants $c_0$ and $c_1$ that if the alphabet size is at least $4$,
the number of sequences is at least $c_1\log n$, the motif length is
at least $c_0\log n$, and each character in motif region has
probability at most ${1\over (\log n)^{2+\mu}}$ of mutation for some
fixed $\mu>0$, then motif can be recovered in $(O({n\over
\sqrt{h}}(\log n)^{7\over 2}+h^2\log^2 n), O(\log n))$ time, where
$n$ is the longest length of any input sequences, and
$h=\min(|G|,n^{2\over 5})$
 The algorithm total time is sublinear if the motif
length $|G|$ is in the range  $[(\log n)^{7+\mu}, {n\over (\log
n)^{1+\mu}}]$. This is the first sublinear time algorithm with
rigorous analysis in this model.

(2) There exists a randomized algorithm such that there are positive
constants $c_0, c_1$, and $\alpha$ that if the alphabet size is at
least $4$, the number of sequences is at least $c_1\log n$, the
motif length is at least $c_0\log n$, and each character in motif
region has probability at most $\alpha$ of mutation, then motif can
be recovered in $(O({n^2\over |G|}(\log n)^{O(1)}), O(\log n))$
time.

(3) There exists a deterministic algorithm such that there are
positive constants $c_0, c_1$, and $\alpha$ that if the alphabet
size is at least $4$, the number of sequences is at least $c_1\log
n$, the motif length is at least $c_0\log n$, and each character in
motif region has probability at most $\alpha$ of mutation, then
motif can be recovered in $(O({n^2}(\log n)^{O(1)}), O(\log n))$
time.

The methods developed in this paper have been used in the software
implementation. We observed some encouraging results that show
improved performance for motif detection compared with other
softwares.
\end{abstract}

%\end{document}
\newpage

\section{Introduction}
Motif discovery is an important problem in computational biology and
computer science. For instance, it has applications in coding
theory~\cite{FrancesLitman97,GasieniecJanssonLingas99}, locating
binding sites and conserved regions in unaligned
sequences~\cite{StormoHartzell91,LawrenceReilly90,HertzStormo94,Stormo90},
genetic drug target identification~\cite{LanctotLiMaWangZhang99},
designing genetic probes \cite{LanctotLiMaWangZhang99}, and
universal PCR primer
design~\cite{LucasBusch91,DopazoSobrino93,ProutskiHolme96,LanctotLiMaWangZhang99}.

This paper focuses on the application of motif discovery to find
conserved regions in a set of given DNA, RNA, or protein sequences.
Such conserved regions may represent common biological functions or
structures. Many performance measures have been proposed for motif
discovery. Let $C$ be a subset of $0$-$1$ sequences of length $n$.
The {\it covering radius} of $C$ is the smallest integer $r$ such
that each vector in $\{0, 1\}^n$ is at a distance at most $r$ from a
string in $C$. The decision problem associated with the covering
radius for a set of binary sequences is
NP-complete~\cite{FrancesLitman97}. The similar closest string and
substring problems were proved to be
NP-hard~\cite{FrancesLitman97,LanctotLiMaWangZhang99}.
%The problem was proved to be NP-hard for some
%formulations~\cite{FrancesLitman97,LanctotLiMaWangZhang99}.
Some approximation algorithms have been proposed.
  Li et
al.~\cite{LiMaWang99} gave an approximation scheme for the closest
string and substring problems. The related consensus patterns
problem is that given $n$ sequences $s_1,\cdots, s_n$, find a region
of length $L$ in each $s_i$, and a string $s$ of length $L$ so that
the total Hamming distance from $s$ to these regions is minimized.
Approximation algorithms for the consensus patterns problem were
reported in~\cite{LiMaWang99b}. Furthermore, a number  of heuristics
and programs have been
developed~\cite{PevznerSze00,KeichPevzner02,KeichPevzner02b,WangDong05,ChinLeung05}.

In many applications,  motifs are faint and may not be apparent when
two sequences alone are compared but may become clearer when  more
sequences are compared together~\cite{Gusfield97}.
%That is,
%``pairwise alignments whispers and multiple alignments shot out".
For this reason, it has been conjectured that comparing more
sequences together can help with identifying faint motifs.
%In this work, we give the first analytical proof for this conjecture.
This paper is a theoretical approach with a rigorous probabilistic
analysis.

We study a natural probabilistic model for motif discovery. In this
model, there are $k$ background sequences and each character in the
background sequence is a random character from an alphabet $\Sigma$.
A motif $G=g_1g_2\ldots g_m$ is a string of $m$ characters. Each
background sequence is implanted a probabilistically generated
approximate copy of  $G$.  For a probabilistically  generated
approximate copy $b_1b_2\ldots b_m$ of $G$, every character $b_i$ is
probabilistically  generated such that the probability for $b_i\neq
g_i$, which is called a {\it mutation}, is at most $\alpha$. This
model was first proposed in \cite{PevznerSze00} and has been widely
used in experimentally testing motif discovery programs
\cite{KeichPevzner02,KeichPevzner02b,WangDong05,ChinLeung05}. We
note that a mutation in our model converts a character $g_i$ in  the
motif into a different character $b_i$ without probability
restriction. This means that a character $g_i$ in the motif may not
become any character $b_i$ in $\Sigma-\{g_i\}$ with equal
probability.

We develop three algorithms that under the probabilistic model, one
can find the implanted motif with high probability via a tradeoff
between computational time and the probability of mutation. Each
algorithm has the preprocessing phase and the voting phase. We use a
pair of function $(t_1(n,k), t_2(n,k))$ to describe the
computational complexity of motif detection algorithm, where $n$ is
the largest length of input sequence, and $k$ is the number of
sequences. Function $t_1(n,k)$ is the time complexity for the part
for preprocessing,  and $t_2(n,k)$ is the time complexity for
recovering one character for motif after preprocessing. The total
time is $O(t_1(n,k)+t_2(n,k)|G|)$.

(1) There exists a randomized algorithm such that there are positive
constants $c_0$ and $c_1$ that if the alphabet size is at least $4$,
the number of sequences is at least $c_1\log n$, the motif length is
at least $c_0\log n$, and each character in motif region has
probability at most ${1\over (\log n)^{2+\mu}}$ of mutation for some
fixed $\mu>0$, then motif can be recovered in $(O({n\over
\sqrt{h}}(\log n)^{7\over 2}+h^2\log^2 n), O(\log n))$ time, where
$n$ is the longest length of any input sequences, and
$h=\min(|G|,n^{2\over 5})$
 The algorithm total time is sublinear if the motif
length $|G|$ is in the range  $[(\log n)^{7+\mu}, {n\over (\log
n)^{1+\mu}}]$. This is the first sublinear time algorithm with
rigorous analysis in this model.

(2) There exists a randomized algorithm such that there are positive
constants $c_0, c_1$, and $\alpha$ that if the alphabet size is at
least $4$, the number of sequences is at least $c_1\log n$, the
motif length is at least $c_0\log n$, and each character in motif
region has probability at most $\alpha$ of mutation, then motif can
be recovered in $(O({n^2\over |G|}(\log n)^{O(1)}), O(\log n))$
time.

(3) There exists a deterministic algorithm such that there are
positive constants $c_0, c_1$, and $\alpha$ that if the alphabet
size is at least $4$, the number of sequences is at least $c_1\log
n$, the motif length is at least $c_0\log n$, and each character in
motif region has probability at most $\alpha$ of mutation, then
motif can be recovered in $(O({n^2}(\log n)^{O(1)}), O(\log n))$
time.

The research in this model has been reported
in~\cite{FuKaoWang09b,FuKaoWang08b,LiuMaWang08}.
In~\cite{FuKaoWang09b}, Fu et al. developed an algorithm that needs
the alphabet size to be a constant that is much larger than 4.
In~\cite{FuKaoWang08b}, our algorithm cannot handle all possible
motif patterns. In~\cite{LiuMaWang08}, Liu et al. designed algorithm
that runs in $O(n^3)$ time and is lack of rigorous analysis about
its performance. The motif recovery in this natural and simple model
has not been fully understood and seems a complicated problem.

This paper presents two new randomized algorithms and one new
deterministic algorithm. They make advancements in the following
aspects: 1. The algorithms are much faster than those before. Our
algorithms can even run in sublinear time. 2. They can handle any
motif pattern. 3. The restriction for the alphabet size is as small
as four, giving them potential applications in practical problems
since gene sequences have an alphabet size $4$. 4. All algorithms
have rigorous proofs about their performances.

%We give a brief description of our main  algorithm \algmnam~in
%Section~\ref{sketch-sec}.
The entire \algmnam~is described in
Section~\ref{algorithm-sec}. We analyze \algma
%and state its performance in our main theorem
%Theorem~\ref{main-theorem}
in Section~\ref{analysis-sec}.
%We conclude the paper with an open
%problem in Section~\ref{open-problem-sec}.

\section{Notations and the Model of Sequence Generation}\label{notation-sec}
For a set $A$, $||A||$ denotes the number of elements in $A$.
$\Sigma$ is an alphabet with $||\Sigma||=t\ge 2$. For an integer
$n\ge 0$, $\Sigma^n$ is the set of sequences of length $n$ with
characters from $\Sigma$. For a sequence $S=a_1a_2\cdots a_n$,
$S[i]$ denotes the character $a_i$, and $S[i,j]$ denotes the
substring $a_i\cdots a_{j}$ for $1\le i\le j\le n$. $|S|$ denotes
the length of the sequence $S$.
%If $S'$ is a substring of $S$, we denote $S-S'$ to be the part of
%$S$ that does not belong to $S'$. In other words, if $S=a_1a_2\cdots
%a_n$ and $S'=S[i,j]$, then $S-S'$ represents the sequence
%$a_1a_2\cdots a_{i-1}a_{j+1}\cdots a_n$.
We use $\emptyset$ to represent the empty sequence, which has length
$0$.

Let $G=g_1g_2\cdots g_m$ be a fixed sequence of $m$ characters. $G$
is the motif to be discovered by our algorithm. A
$\Theta_{\alpha}(n,G)$-sequence has the form $S=a_1\cdots
a_{n_1}b_1\cdots b_m a_{n_1+1}\cdots a_{n_2}$, where $n_2+m\le n$,
each $a_i$ has probability ${1\over t}$ to be equal to $\pi$ for
each $\pi\in\Sigma$, and $b_i$ has probability at most $\alpha$ not
equal to $g_i$ for $1\le i\le m$, where $m=|G|$. $\aleph(S)$ denotes
the motif region $b_1\cdots b_m$ of $S$. A mutation converts a
character $g_i$ in  the motif into an arbitrary different character
$b_i$ without probability restriction. This allows a character $g_i$
in the motif to change into any character $b_i$ in $\Sigma-\{g_i\}$
with even different probability. The motif region of $S$ may start
at an arbitrary or worst-case position in $S$. Also, a mutation may
convert a character $g_i$ in the motif into an arbitrary or
worst-case different character $b_i$ only subject to the restriction
that $g_i$ will mutate with probability at most $\alpha$.

%Let $G=g_1g_2\cdots g_m$ be a fixed sequence of $m$ characters. $G$
%is the motif to be discovered by our algorithm.
A $\Psi(n,G)$-sequence has the form $S=a_1\cdots a_{n_1}b_1\cdots
b_m a_{n_1+1}\cdots a_{n_2}$, where $n_2+m\le n$, each $a_i$ has
probability ${1\over t}$ to be equal to $\pi$ for each
$\pi\in\Sigma$, and there are at most $O(1)$ characters $b_i$ not
equal to $g_i$ for $1\le i\le m$ and each mutation occurs at a
random position of $G$, where $m=|G|$.

%$\aleph(S)$ denotes the motif region $b_1\cdots b_m$ of $S$.!!!(this
%sentence is similar to another sentence in the former paragraph)
%fixed.

%Let $S_1, S_2,\cdots, S_k$ be a series of
%$\Theta(n,G,\alpha)$-sequences. Each $S_i$ has an area
%$\aleph(S_i)$ that is considered as a common area $G$ among all
%sequences. Each character in $\aleph(S_i)$ has probability at most
%$\alpha$ to mutate. All the other characters are random symbols in
%$\Sigma$.

For two sequences $S_1=a_1\cdots a_m$ and $S_2=b_1\cdots b_m$ of the
same length, let the {\it relative Hamming distance} $\diff(S_1,
S_2)=\frac{|\{i| a_i\not= b_i(i=1,\cdots, m)\}|}{m}$.

\begin{definition}\label{shift-def}
For two intervals $[i_1, j_1]$ and $[i_2, j_2]$, define
$\shift([i_1, j_1], [i_2, j_2])=\min(|i_1-i_2|, |j_1-j_2|)$.
%Assume that $S=a_1a_2\cdots a_n$ is a sequence. For its substrings
%$S[i_1, j_1]$ and $S[i_2,j_2]$, define $\shift_S(S[i_1,j_1],
%S[i_2,j_2])=\shift([i_1,j_1],[i_2,j_2])$.
%Note that the two
%substring s $S[i_1,j_1]$ and $S[i_2,j_2]$ contain the address
%information in $S$ for computing $\shift()$.
%\min(|i_1-i_2|, |j_1-j_2|)$.
\end{definition}

%\section{Brief Description of Algorithm}
\section{Brief Introduction to Algorithm}

Every detection algorithm in this paper has two phases. The first
phase is preprocessing so that the motif regions from multiple
sequences can be aligned in the same column region. The second phase
is to recover the motif via voting. We use a pair of function
$(t_1(n,k), t_2(n,k))$ to describe the computational complexity of
motif detection algorithm. Function $t_1(n,k)$ is the time
complexity for the preprocessing phase and $t_2(n,k)$ is the time
complexity for outputting one character for motif in the voting
phase.

The motif $G$ is a pattern unknown to algorithm \algmnam, and
algorithm \algmnam~will attempt to recover $G$ from a series of
$\Theta(n,G,\alpha)$-sequences generated by the probabilistic model.

%Recall that a sequence $S$ is generated in this model as follows
%(1). Generate a sequence $S'$ with $n-|G|$ characters, in which each
%character is a random character $\Sigma$. (2). Generate $G'$ such
%that with probability at most $\alpha$, $G'[i]\not=G[i]$. For
%$G'[i]\not=G[i]$, it represents a mutation. A mutation may create an
%arbitrary or worst-case $G'[i]$, with no probability restriction
%except that the mutation occurs with probability at most $\alpha$.
% (3). Insert $G'$, which serves as the motif region
%$\aleph(S)$ of $S$, into any arbitrary or worst-case position of $S'$.

\subsection{Algorithm}

The algorithm first detects a position that is close to the left
motif boundary in a sequence. It finds such a position via sampling
and collision between two sequences. After the rough left boundary a
sequence is found, it is used to find the rough boundaries of the
rest of the sequences. Similarly, we find those right boundaries of
motif among the input sequences. The exact left boundary of each
motif region will be detected in the next phase via voting. Each
character of the motif is recovered by voting among all the
characters at the same positions in the motif regions of input
sequences.

{\bf Descriptions of Algorithm}

Input: $Z=Z_1\cup Z_2$, where $Z_1=\{S_1',\cdots, S_{2k_1}'\}$ and
$Z_2=\{S_1'',\cdots, S_{k_2}''\}$ are two sets of input sequences.

Output:Planted motif in each sequence and consensus string

{\bf Start:}

Randomly select sample points from each sequence both in $Z_1$ and $ Z_2$

For each pair of sequences selected from $Z_1$ and $ Z_2$,

\qquad  Find the rough left and rough right boundaries.

\qquad  Improve rough boundaries.

If motif boundaries of each sequence in $Z_2$ are not empty,

\qquad  Use Voting algorithm to get the planted motifs.

{\bf End of Algorithm}

 \subsection{An Example}

We provide the following example for the brief idea of our
algorithm. Let the following input strings be defined as below. We
assume that the original motif is {\bf TTTTTAACGATTAGCS}. The motif
part is displayed with bold font, and the mutation characters in the
motif region are displayed with small font.

\subsubsection{Input Sequences}

It contains two groups $Z_1=\{S_1',S_2'\}$ and
$Z_2=\{S_1'',S_2'',S_3'', S_4'', S_5''\}$.

\begin{eqnarray*}
Z_1:&&\\
 S_1'&=&GTACCATGGA{\Large\bf TT{\small
 A}TTAACGATTAGCS}TAGAGGACCTA.\\
S_2'&=&AATCCTTA{\Large\bf {\small C}TTTTAACGATTAGCS}GTC.
\end{eqnarray*}

The above two strings are used to detect the initial motif region
and use them to deal with the motif in the second group below.

%\begin{eqnarray}

\begin{eqnarray*}
Z_2:&&\\
 S_1''&=&ATTCGATCCAG{\Large\bf
TTTTTAACG{\small{G}}TTAGCS}CAATTACTTAG.\\
S_2''&=&GCATTGCAT{\Large\bf
TTTTTAACGATTA{\small{C}}CS}GTACTTAGCTAGATC.\\
S_3''&=&TCAGGGCATCGAGAC{\Large\bf
TTTTTA{\small{G}}CGATTAGCS}CTAGAATCAGACCT.\\
S_4''&=&GTACCTGGCATTGAACG{\Large\bf
TTTTTAACGATTAGC{\small{A}}}TGCAGATGGACCTTTA.\\
S_5''&=&AATGGATCAGA{\Large\bf TTTTTAACGATT{\small{C}}GCS}CTAGATTCAG.
\end{eqnarray*}

\subsubsection{Select Sample Points}

Some sample points of two sequences in $Z_1$ are selected and
marked.
\begin{eqnarray*}
 S_1'&=&GT\dot{A}CC\dot{A}TG\dot{G}A{\Large\bf T\dot{T}{\small A}{T}TA\dot{A}CGA\dot{T}T\dot{A}GC\dot{S}}TA\dot{G}AG\dot{G}ACC\dot{T}A.\\
S_2'&=&\dot{A}AT\dot{C}CTT\dot{A}{\Large\bf {\small
C}\dot{T}TTT\dot{A}AC\dot{G}A\dot{T}TA\dot{G}CS}\dot{G}TC.
\end{eqnarray*}

\subsubsection{Collision Detection}

In this step, the left and right rough boundaries of two sequences will be marked.
The following show the left collision, which happens nearby the left
motif boundary and are marked by two overline $\overline{TATT}$ and
$\overline{TTTT}$ subsequences.
\begin{eqnarray*}
 S_1'&=&GT\dot{A}CC\dot{A}TG\dot{G}A{\Large\bf T\overline{\dot{T}ATT} A\dot{A}CGA\dot{T}T\dot{A}GC\dot{S}}TA\dot{G}AG\dot{G}ACC\dot{T}A.\\
S_2'&=&\dot{A}AT\dot{C}CTT\dot{A}{\Large\bf {\small
C}\overline{\dot{T}TTT}\dot{A}AC\dot{G}A\dot{T}TA\dot{G}CS}\dot{G}TC.
\end{eqnarray*}

The following show the right collision, which happens nearby the
right motif boundary and are marked by two overline
$\overline{TTAG}$ subsequences.
\begin{eqnarray*}
 S_1'&=&GT\dot{A}CC\dot{A}TG\dot{G}A{\Large\bf T\dot{T}ATT A\dot{A}CGA\overline{\dot{T}T\dot{A}G}C\dot{S}}TA\dot{G}AG\dot{G}ACC\dot{T}A.\\
S_2'&=&\dot{A}AT\dot{C}CTT\dot{A}{\Large\bf {\small
C}\dot{T}TTT\dot{A}AC\dot{G}A\overline{\dot{T}TA\dot{G}}CS}\dot{G}TC.
\end{eqnarray*}

\subsubsection{Improving the Boundaries}

In the early phase of the algorithm, we first detect a small piece
of motif in $S_1'$ by comparing $S_1'$ and $S_2'$. Assume ``T{\small
A}TT" and ``TTAG" are found in the left and right motif region of
$S_1'$ respectively. The rough motif length will be calculated via
the difference of the location first character `T' of the first
subsequence and the location of the last character `G' of the
second subsequence. The position marked by ``\underline{A}" is the
rough left boundary of motif and the position marked by
``\underline{T}" is the rough right boundary of motif in $S_1'$
below.
\begin{eqnarray*}
 S_1'&=&GTACCATGG\underline{A}{\Large\bf TTATTAACGATTAGCS}\underline{T}AGAGGACCTA.\\
S_2'&=&AATCCTT\underline{A}{\Large\bf {\small
C}TTTTAACGATTAGCS}\underline{G}TC.
\end{eqnarray*}

\subsubsection{Select Sample Points for the Sequences in $Z_2$}

Some sample points near the motif boundaries of $S_1'$ are selected.

 $S_1''=GTACCATG\dot{G}A{\Large\bf T\dot{T}A\dot{T}T
AACGATT\dot{A}G\dot{C}S}T\dot{A}GAGGACCTA$.

Sample points are selected in each sequence in $Z_2$.

\begin{eqnarray*}
S_1''&=&A\dot{T}TC\dot{G}ATCC\dot{A}G{\Large\bf
T\dot{T}T\dot{T}TAACGGTTAG\dot{C}S}C\dot{A}AT\dot{T}ACTT\dot{A}G.\\
S_2''&=&G\dot{C}ATT\dot{G}CAT{\Large\bf
T\dot{T}TTTAACGATTAC\dot{C}S}GT\dot{A}CTT\dot{A}GCT\dot{A}GA\dot{T}C.\\
S_3''&=&\dot{T}CA\dot{G}GGCA\dot{T}CGA\dot{G}AC{\Large\bf TTT\dot{T}TAGCGATTAG\dot{C}S}CTA\dot{G}AATC\dot{A}GAC\dot{C}T.\\
S_4''&=&GT\dot{A}CCT\dot{G}GCAT\dot{T}GAACG{\Large\bf T\dot{T}TTTAACGATT\dot{A}GCA}TGC\dot{A}GAT\dot{G}GACCT\dot{T}TA.\\
S_5''&=&AA\dot{T}GGA\dot{T}CAGA{\Large\bf
T\dot{T}TTTAACGATTCG\dot{C}S}CTA\dot{G}ATT\dot{C}AG.
\end{eqnarray*}

\subsubsection{Collision Detection Between $S_1'$ with the Sequences in $Z_2$}

Some sample points near the motif boundaries of $S_1'$ are selected.

 $S_1''=GTACCATG\dot{G}A{\Large\bf T\dot{T}A\dot{T}T
AACGATT\dot{A}G\dot{C}S}T\dot{A}GAGGACCTA$.

Sample points are selected in each sequence in $Z_2$.

\begin{eqnarray*}
S_1''&=&A\dot{T}TC\dot{G}ATCC\dot{A}G{\Large\bf
T\overline{\dot{T}T\dot{T}T}AACGGT\overline{TAG\dot{C}}S}C\dot{A}AT\dot{T}ACTT\dot{A}G.\\
S_2''&=&G\dot{C}ATT\dot{G}CAT{\Large\bf
T\overline{\dot{T}TTT}AACGAT\overline{TAC\dot{C}}S}GT\dot{A}CTT\dot{A}GCT\dot{A}GA\dot{T}C.\\
S_3''&=&\dot{T}CA\dot{G}GGCA\dot{T}CGA\dot{G}AC{\Large\bf TTT\overline{\dot{T}TAG}CGAT\overline{TAG\dot{C}}S}CTA\dot{G}AATC\dot{A}GAC\dot{C}T.\\
S_4''&=&GT\dot{A}CCT\dot{G}GCAT\dot{T}GAACG{\Large\bf T\overline{\dot{T}TTT}AACG\overline{ATT\dot{A}}GCA}TGC\dot{A}GAT\dot{G}GACCT\dot{T}TA.\\
S_5''&=&AA\dot{T}GGA\dot{T}CAGA{\Large\bf
T\overline{\dot{T}TTT}AACGAT\overline{TCG\dot{C}}S}CTA\dot{G}ATT\dot{C}AG.
\end{eqnarray*}

\subsubsection{Improving the Motif Boundaries for the Sequences in $Z_2$}

After the collision with the sequences in $Z_2$, we obtain the rough
location of motifs of the sequences in $Z_2$. Their motif boundaries
for the sequences in $Z_2$ are improved.

$S_1''=GTACCATGGA{\Large\bf \overline{TTAT}T
AACGATT\overline{AGCS}}TAGAGGACCTA$.

The improved motif boundaries of the sequences in $Z_2$ are marked
below.

\begin{eqnarray*}
S_1''&=&ATTCGATCCA\underline{G}{\Large\bf
TTTTTAACG{\small{G}}TTAGCS}\underline{C}AATTACTTAG.\\
S_2''&=&GCATTGC\underline{A}T{\Large\bf
TTTTTAACGATTACCS}\underline{G}TACTTAGCTAGATC.\\
S_3''&=&TCAGGGCATCGAGA\underline{C}{\Large\bf TTTTTAGCGATTAGCS}\underline{C}TAGAATCAGACCT.\\
S_4''&=&GTACCTGGCATTGAAC\underline{G}{\Large\bf TTTTTAACGATTAGCA}\underline{T}GCAGATGGACCTTTA.\\
S_5''&=&AATGGATCAG\underline{A}{\Large\bf
TTTTTAACGATTCGCS}C\underline{T}AGATTCAG.
\end{eqnarray*}

\subsubsection{Motif Boundaries for the Sequences in $Z_2$}

$S_1''=GTACCATGGA{\Large\bf \overline{TTAT}T
AACGATT\overline{AGCS}}TAGAGGACCTA$.

Use the pair $(G_L,G_R)$ with $G_L={\Large\bf \overline{TTAT}}$ and
$G_R={\Large\bf\overline{AGCS}}$ to find the motif boundaries in the
sequences of $Z_2$. The rough boundaries of the second group is
marked below with underlines.

%!!!($R_R$ should be $G_R$?)
%Fixed.

\begin{eqnarray*}
S_1''&=&ATTCGATCCA\underline{G}{\Large\bf
TTTTTAACG{\small{G}}TTAGCS}\underline{C}AATTACTTAG.\\
S_2''&=&GCATTGCA\underline{T}{\Large\bf
TTTTTAACGATTACCS}\underline{G}TACTTAGCTAGATC.\\
S_3''&=&TCAGGGCATCGAGA\underline{C}{\Large\bf TTTTTAGCGATTAGCS}\underline{C}TAGAATCAGACCT.\\
S_4''&=&GTACCTGGCATTGAAC\underline{G}{\Large\bf TTTTTAACGATTAGCA}\underline{T}GCAGATGGACCTTTA.\\
S_5''&=&AATGGATCAG\underline{A}{\Large\bf
TTTTTAACGATTCGCS}\underline{C}TAGATTCAG.
\end{eqnarray*}

\subsubsection{Extracting the Motif Regions}

The motif regions of the second group will be extracted. The
original motif is recovered via voting at each column.
\begin{eqnarray*}
G_1''&=&{\Large\bf TTTTTAACG{\small{G} }TTAGCS}\\
G_2''&=&{\Large\bf TTTTTAACGATTA{\small{C}}CS}\\
G_3''&=&{\Large\bf TTTTTA{\small{G}}CGATTAGCS}\\
G_4''&=&{\Large\bf TTTTTAACGATTAGC{\small{A}}}\\
G_5''&=&{\Large\bf TTTTTAACGATT{\small{C}}GCS}
\end{eqnarray*}

\subsubsection{Recovering Motif via Voting}

The original motif {\bf TTTTTAACGATTAGCS} is recovered via voting at
all columns. For example, the last {\bf S} in the motif is recovered
via voting among the characters {\bf S, S, S, A, S} in the last
column.

\subsection{Our Results}

We give an algorithm for the case with at most ${1\over (\log
n)^{2+\mu}}$ mutation rate. The performance of the algorithm is
stated in Theorem~\ref{main-theorem1}. Theorem~\ref{main-theorem1}
implies Corollary~\ref{corollary1} by selecting $k=c_1\log n$ with
some constant $c_1$ large enough.

\begin{theorem}\label{main-theorem1}Assume that $\mu$ is a fixed number in $(0,1)$
and the alphabet size $t$ is at least $4$. There exists a randomized
algorithm such that there is a constant $c_0$ that
 if the length of the motif $G$ is at least $c_0\log n$,
then given $k$ independent $\Theta(n,G,{1\over (\log
n)^{2+\mu}})$-sequences, the algorithm outputs $G'$ such that

1) with probability at most $e^{-\Omega(k)}$, $|G'|\not=|G|$, and

2) for each $1\le i\le |G|$, with probability at most
$e^{-\Omega(k)}$, $G'[i]\not=G[i]$, and

3) with probability at most ${k\over n^3}$, the algorithm
\algmnam~does not stop in $(O(k({n\over \sqrt{h}}(\log n)^{5\over
2}+h^2\log n)), O(k))$ time,

 where $n$ is the longest length of any input
sequences, and $h=\min(|G|,n^{2\over 5})$.
\end{theorem}

\begin{corollary}\label{corollary1}
There exists a randomized algorithm such that there are positive
constants $c_0, c_1$ and $\mu$ that if the alphabet size is at least
$4$, the number of sequences is at least $c_1\log n$, the motif
length is at least $c_0\log n$, and each character in motif region
has probability at most ${1\over (\log n)^{2+\mu}}$ of mutation,
then motif can be recovered in $(O({n\over \sqrt{h}}(\log n)^{7\over
2}+h^2\log^2 n), O(\log n))$ time, where $n$ is the longest length
of any input sequences, and $h=\min(|G|,n^{2\over 5})$.
\end{corollary}

We give a randomized algorithm for the case with $\Omega(1)$
mutation rate. The performance of the algorithm is stated in
Theorem~\ref{main-theorem2}. Theorem~\ref{main-theorem2} implies
Corollary~\ref{corollary2} by selecting $k=c_1\log n$ with some
constant $c_1$ large enough..

\begin{theorem}\label{main-theorem2}Assume that the alphabet size $t$ is at least $4$. There
exists a randomized algorithm such that there is a constant $c_0$
that
 if the length of the motif $G$ is at least $c_0\log n$,
then given $k$ independent $\Theta(n,G,\mu))$-sequences, the
algorithm
 outputs
$G'$ such that

1) with probability at most $e^{-\Omega(k)}$, $|G'|\not=|G|$, and

2) for each $1\le i\le |G|$, with probability at most
$e^{-\Omega(k)}$, $G'[i]\not=G[i]$,

3)  with probability at most ${k\over n^3}$, the algorithm
\algmnam~does not stop in $(O(k({n^2\over |G|}(\log n)^{O(1)}+h^2)),
O(k))$,

 where $n$ is the longest length of any input
sequences, and $h=\min(|G|,n^{2\over 5})$.
\end{theorem}

\begin{corollary}\label{corollary2}
There exists a randomized algorithm such that there are positive
constants $c_0, c_1$, and $\alpha$ that if the alphabet size is at
least $4$, the number of sequences is at least $c_1\log n$, the
motif length is at least $c_0\log n$, and each character in motif
region has probability at most $\alpha$ of mutation, then motif can
be recovered in $(O({n^2\over |G|}(\log n)^{O(1)}), O(\log n))$
time.
\end{corollary}

We give a deterministic algorithm for the case with $\Omega(1)$
mutation rate.  The performance of the algorithm is stated in
Theorem~\ref{main-theorem3}. Theorem~\ref{main-theorem3} implies
Corollary~\ref{corollary3} by selecting $k=c_1\log n$ with some
constant $c_1$ large enough.

\begin{theorem}\label{main-theorem3}Assume that the alphabet size $t$ is at least $4$. There
exists a deterministic
%randomized!!!(should be deterministic?) %Fixed.
algorithm such that there is a constant $c_0$ that if the length of
the motif $G$ is at least $c_0\log n$, then given $k$ independent
$\Theta(n,G,\mu))$-sequences, algorithm runs in $(O(n^2(\log
n)^{O(1)}+h^2k), O(k))$, and outputs $G'$ such that

1) with probability at most $e^{-\Omega(k)}$, $|G'|\not=|G|$, and

2) for each $1\le i\le |G|$, with probability at most
$e^{-\Omega(k)}$, $G'[i]\not=G[i]$,

3)  with probability at most ${k\over n^3}$, the algorithm
\algmnam~does not stop in $(O(k(n^2(\log n)^{O(1)}+h^2)), O(k))$
time,

 where $n$ is the longest length of any input
sequences, and $h=\min(|G|,n^{2\over 5})$.
\end{theorem}

\begin{corollary}\label{corollary3}
There exists a deterministic algorithm such that there are positive
constants $c_0, c_1$, and $\alpha$ that if the alphabet size is at
least $4$, the number of sequences is at least $c_1\log n$, the
motif length is at least $c_0\log n$, and each character in motif
region has probability at most $\alpha$ of mutation, then motif can
be recovered in $(O({n^2}(\log n)^{O(1)}), O(\log n))$ time.
\end{corollary}

%\section{Algorithm for Large Alphabet} %Detecting Any Motif}
\section{\algmname}
In this section, we give an unified approach to describe three
algorithms. %which is be converted into three different algorithms.
The performance of the algorithms is stated in the
Theorems~\ref{main-theorem1}, \ref{main-theorem2}, and
\ref{main-theorem3}. The description of \algma is given at
section~\ref{algorithm-sec}. The analysis of the algorithm is given
at section~\ref{analysis-sec}.

%\begin{theorem}\label{main-theorem}Assume that $\tau $ and $\mu$ are fixed numbers in $(0,1)$
%and the alphabet size $t$ is at least $4$. Then there exists a
%constant $c_0$ such that
% if the length of the motif $G$ is at least $c_0\log n$,
%then given $k$ independent $\Theta(n,G,o({1\over (\log
%n)^{2+\mu}}))$-sequences, algorithm \algmnam~has time complexity
%$(O({n\over \sqrt{h}}(\log n)^{3\over 2}+h^2, O(k))$, and outputs
%$G$ with probability at most $e^{-\Omega(k)}$ to be incorrect at
%each character, where $n$ is the longest length of any input
%sequences, and $h=\min(|G|,n^{2\over 5})$.
%\end{theorem}

\subsection{Some Parameters}
\begin{definition}\label{param-def}\scrod
\begin{enumerate}
\item
Constant $x$ is selected to be $10$. This parameter controls the
failure probability of our algorithms to be at most ${1\over 2^x}$.
\item
The size of alphabet is $t$ that is at least $4$.
\item
Select a constant $\rho_0\in (0,1)$ to have
inequality~(\ref{alpha-init-ineqn})
\begin{eqnarray}
\rho_0<{t-1\over 2t}.\label{alpha-init-ineqn}
\end{eqnarray}
\item The constant $\epsilon\in (0,1)$ is selected to satisfy
\begin{eqnarray}
\epsilon< \min(({t-1\over t}-(2\rho_0+2\epsilon)),\  \ {1\over
5}(1-{2\over t-1}-{4\over 2^x}),\ \  {1\over
3}).\label{epsilon-set-ineqn}
\end{eqnarray}
The existence of $\epsilon$ follows from
inequality~(\ref{alpha-init-ineqn}). The constant $\epsilon$ is used
to control the mutation in the motif area. It is a part of parameter
$\beta$ defined in item~(\ref{beta-def}) of this definition.

\item
Let $c=e^{-{\epsilon^2\over 3}}$. The constant $c$ is used to simply
probabilistic bounds which are derived from the applications of
Chernoff bounds (See Corollary~\ref{chernoff-lemma-a}).

\item
Define $r(y)=({1\over t-1}+{c^y\over 1-c})$.

\item
 Define $u_1$ to be a large constant that for all $v\ge 0$,
\begin{eqnarray}
{2(v+u_1)c^{v+u_1}\over (1-c)^2}\le {1\over 5\cdot
2^x}.\label{v-set-ineqn2}
\end{eqnarray}

\item Select constant $\rho_1\in (0,1)$ such that
\begin{eqnarray}
{2\over t-1}+{4\over 2^x}+5\epsilon+\rho_1<1.\label{rho1-ineqn}
\end{eqnarray}

The existence of $\rho_1$ follows from $\epsilon< {1\over
5}(1-{2\over t-1}-{4\over 2^x})$, which is implied by
inequality~(\ref{epsilon-set-ineqn}).

\item Select constant $\rho_2\in (0,1)$ and  constant positive integer $v$ large enough such that
\begin{eqnarray}
&&{6(v+u_1)c^v\over 1-c}+\rho_2<\rho_1,\ \ {\rm and}\ \label{rho2-ineqn}\\
&&({1\over 2^x}+(v+u_1){c^{v}\over 1-c}+{c^v\over 1-c}+{1\over
5\cdot 2^x})\le {1/2}. \label{v2-ineqn}
\end{eqnarray}

\item
Define $\varsigma_0={1\over 2^x}$, and
$\varphi(v)=(v+u_1)({c^{v}\over 1-c}+{c^{v}\over 1-c})$.
%!!!(miss open parenthesis) %Fixed.

%The inequality $1-2(Q_0+V_0+(R+2\epsilon))-(\alpha+\epsilon)>0$
%holds, where $V_0=(2(\varsigma_2(n)+ (v+u_1){c^{v+u_2}\over
%1-c}+{c^{v}\over 1-c})+\epsilon)$.

\item Select constant $\alpha_0$ such that
\begin{eqnarray}
4(v-1)\alpha_0+\alpha_0&<&\rho_2,\ \ {\rm and}\label{rho2-alpha0-inequality}\\
\alpha_0&<&\rho_0.
\end{eqnarray}
 Adding inequalities~(\ref{rho1-ineqn}), (\ref{rho2-ineqn}), and (\ref{rho2-alpha0-inequality}), we have
 inequality~(\ref{median-ineqn})

\begin{eqnarray}
({2\over t-1}+{4\over 2^x}+5\epsilon)+{6(v+u_1)c^v\over
1-c}+(4(v-1)\alpha_0+\alpha_0)<1.\label{median-ineqn}
\end{eqnarray}

By arranging the terms in inequality~(\ref{median-ineqn}) and the
definitions of $r(v)$ and $\varphi(v)$, we have
inequality~(\ref{v-alpha0-ineqn})

\begin{eqnarray}
2((2(v-1)\alpha_0+{c^v\over
1-c})+r(v)+2(\varsigma_0+\varphi(v))+2\epsilon)+(\alpha_0+\epsilon)<1.\label{v-alpha0-ineqn}
\end{eqnarray}

\item
The maximal mutation rate $\alpha$ for the second algorithm
(Theorem~\ref{main-theorem2}) and third algorithm
(Theorem~\ref{main-theorem3}) are selected as $\alpha_0$. Since the
mutation rate of our sublinear time algorithm is bounded by ${1\over
(\log n)^{2+\mu}}$, the maximal mutation rate $\alpha$ for the first
algorithm (Theorem~\ref{main-theorem1}) is less than $\alpha_0$ when
$n$ is large enough. We always assume that all mutation rates
$\alpha$ in our three algorithms are in the range $(0,\alpha_0]$.

\item
Define $q(y)=2(v-1)\alpha+{2c^y\over 1-c}$. %(1-\alpha)^{2(y-1)}-{2c^y\over 1-c}$.
By inequality~(\ref{v-alpha0-ineqn}), the definition of $q(y)$, and
the fact $\alpha\in (0,\alpha_0)$, we have
\begin{eqnarray}
2(q(v)+r(v)+2(\varsigma_0+\varphi(v))+2\epsilon)+(\alpha_0+\epsilon)&<&1.\label{v-set-ineqn1}
\end{eqnarray}

Inequality~(\ref{v-set-ineqn1}) implies $q(v)\le {1\over 2}$. By
inequality~(\ref{v2-ineqn}), we have that
\begin{eqnarray}
 ({1\over 2^x}+(v+u_1){c^{v}\over
1-c}+{c^v\over 1-c}+{1\over 5\cdot 2^x})+q(v)\le
{3/4}\label{support-P0-Q0-inequality}
\end{eqnarray}

\item\label{beta-def}
 Let
$\beta=2\alpha+2\epsilon$. The parameter $\beta$ controls the
similarity of $\aleph(S)$ and the original motif $G$ (see
Lemma~\ref{base}).

\item
Define $R=r(v)$.

%\item
%Define $u_0$ to satisfy  ${2(v+u_1)c^{v+u_1}\over (1-c)^2}{1\over
%5???2^x}$ and ????.

\item
We define the following $Q_0$. %It will be first used in Lemma~\ref{base}.
\begin{eqnarray}
 Q_0=q(v).\label{Q0-def-eqn}
\end{eqnarray}
The parameter $Q_0$ used in~Lemma~\ref{base} gives an upper  bound
of the probability that a $\Theta(n,G,\alpha)$-sequence $S$ whose
$\aleph(S)$ will not be similar enough to the original motif $G$
according to the conditions in~Lemma~\ref{base}.

\item
Select constant $d_0$ such that
\begin{eqnarray}
n^3c^{d_0\log n}<{1\over 5\cdot 2^x}.\label{d0-sel-eqn}
\end{eqnarray}

\item
Select constant $d_1$ such that $(v+u_1)c^{d_1\log n}<{1\over 5\cdot
2^x}$.

\item
Select number $u_2$ such that
\begin{eqnarray}
(d_1\log n)(v+u_1){c^{v+u_2}\over 1-c}&\le& {1\over 5\cdot 2^x}.\ {\rm and} \ \\
(v+u_1){c^{v+u_2}\over 1-c}&<&{1\over 5\cdot 2^x}
\end{eqnarray}
Since only $n$ is variable, we can make $u_2=O(\log\log n)$.

\item
For a fixed $c\in (0,1)$, define $\delta_c={\ln{1\over c}\over 2}$.
\end{enumerate}
\end{definition}

\subsection{Description of  \algma}\label{algorithm-sec}

The algorithm is described in this section. Before presenting the
algorithm, we define some notions.

\begin{definition}\label{match-def}{\scrod}
\begin{itemize}
\item
Two sequences $X_1$ and $X_2$ are {\it weak left matched} if (1)
both $|X_1|$ and $|X_2|$ are at least $d_0\log n$, (2)
$\diff(X_1[1,i], X_2[1,i])\le \beta$ for all integers $i$, $v\le
i\le d_0\log n$.

\item
Two sequences $X_1$ and $X_2$ are {\it left matched} if (1) $d_0\log
n\le |X_1|,|X_2|$, (2) $X_1[i]=X_2[i]$ for $i=1,\cdots, v-1$, and
(3) $\diff(X_1[1,i], X_2[1,i])\le \beta$ for all integers $i$, $v\le
i\le d_0\log n$.

\item
Two sequences $X_1$ and $X_2$ are {\it weak right matched} if
$X_1^R$ and $X_2^R$ are weak left matched, where $X^R=a_n\cdots a_1$
is the inverse sequence of $X=a_1\cdots a_n$.

\item
Two sequences $X_1$ and $X_2$ are {\it right matched} if $X_1^R$ and
$X_2^R$ are left matched, where $X^R=a_n\cdots a_1$ is the inverse
sequence of $X=a_1\cdots a_n$.

\item
Two sequences $X_1$ and $X_2$ are {\it matched} if $X_1$ and $X_2$
are both left and right matched.
\end{itemize}
\end{definition}

%If the motif length satisfies the condition of
%Theorem~\ref{main-theorem},
Variable $L$ will be controlled in the range $L\in [(\log
n)^{3+\epsilon_1}, n^{{2\over 5}-\epsilon_2}]$ in our algorithm with
high probability. We define the following functions that depend on
$L$.

\begin{definition}\label{M-M1-def}
Define $M(L)={\sqrt{3\log n+x}\over \sqrt{1-\gamma}}\sqrt{L}\log n$.
 Define $M_1(L)={\delta_{c_0} M(L)\over \log n}$ (see
Definition~\ref{param-def} for $\delta_c$), where $c_0={1\over 4}$.
\end{definition}

%In the rest of this paper, when $L$ is fixed, we just use $M$ to
%represent $M(L)$, and $M_1$ to represent $M_1(L)$.
We would like to
minimize the function $({n\over L}M+L^2)\log n$. This selection can
make the total time complexity sublinear.

%(P.S.??? Please think a way to approximate the motif length first.
%Then we can find an efficient way to detect the motif with the known
%$L$. We may start at a point that minimizes the function $({n\over
%L}M+L^2)\log n$, and shrink $L$ if it does not work. It looks that
%$L=n^{2\over 5}$ is a starting point. It follows from $1-{x\over
%2}=2x$.)

% \vskip 10pt

%{\bf LoadInputSequences()}

%{\bf Steps:}

%\qquad  Independently generate $k$ $\Theta(n,G,\alpha)$-sequences
%$S_1, S_2,\cdots, S_k$.

%\qquad \qquad Let $Z=\{S_1,\cdots, S_{k}\}$.

%\qquad Return $Z$.

%{\bf End of LoadInputSequences}

%\vskip 10pt
 %{\bf \phaseone~ of the \algma}

\begin{definition}
For a $\Theta_{\alpha}(n,G)$ sequence $S$, define \LB($S$) to be the
left boundary $l$ of the motif region $\aleph(S)$ in $S$, and
\RB($S$) to be the right boundary $r$ of the motif region
$\aleph(S)$ in $S$ such that $\aleph(S)=S[l,r]$.
\end{definition}

\subsubsection{\phaseone~of \algma}
 The first phase of \algma~finds the rough motif boundaries of all
input sequences. It first detects the rough motif boundaries of one
sequence via comparing two input sequences. Then the rough
boundaries of the first sequence is used to find the rough motif
boundaries of other input sequences.

Three algorithms share most of the functions. We have a unified
approach to describe them. A special variable ``\algtype" selects
one of the three algorithms, respectively.

\begin{definition}
Let \algtype~represent one of the three algorithm types,
``\sublinear", "\randomized", and "\deterministic".
\end{definition}

\begin{definition}
Assume that $A_1$ is a set of positions in a $\Theta_{\alpha}(n,G)$
sequence $S_1$ and $A_2$ is a set of positions in a
$\Theta_{\alpha}(n,G)$ sequence $S_2$. If there is a position
$a_1\in A_1$ and $a_2\in A_2$ such that for some position $j$ with
$1\le j\le |G|$, $a_1$ is the position of $\aleph(S_1)[j]$ in $S_1$
and $a_2$ is the position of $\aleph(S_2)[j]$ in $S_2$, then $A_1$
and $A_2$ have a {\it collision} at $(a_1,a_2)$.
\end{definition}

In the following function Collision-Detection, the parameter
$\omega\le \beta$ is defined below in the three  algorithms.

\begin{eqnarray}
\omega_{\mbox{\algtype}}&=& \left\{ \begin{array}{ll}
                0 & \mbox{\ if \ } \mbox{\algtype=\sublinear;}\\
                \beta & \mbox{\ if \ }\mbox{\algtype=\randomized;}\\
                 \beta & \mbox{\ if \ }\mbox{\algtype=\deterministic.}\\
          \end{array}
          \right.
\end{eqnarray}

\vskip 10pt {\bf Collision-Detection$(S_1, U_1, S_2, U_2)$}

{\bf Input:} a pair of $\Theta(n,G,\alpha)$-sequences $S_1$ and
$S_2$, $U_i$ is a set of locations in $S_i$ for $i=1,2$.

{\bf Output:} the left and right rough boundaries of two sequences.

\qquad Let $D_1$ be all subsequences $S_1[a,a+d_0\log n-1]$ of $S_1$
of length $d_0\log n$ with $a\in U_1$.

\qquad Let $D_2$ be all subsequences $S_2[b,b+d_0\log n-1]$ of $S_2$
of length $d_0\log n$ with $b\in U_2$.

%\qquad Sort all subsequences in $D_1\cup D_2$ via bucket sorting.

\qquad Find two  subsequences $X_1=S_1[a_1,a_1+d_0\log n-1]\in D_1$
and

\qquad $X_2=S_2[b_1,b_1+d_0\log n-1]\in D_2$ such that $a_1$ is the
least and $\diff(X_1,X_2)\le \omega_{\algtype}$.

\qquad Find two subsequences $X_1'=S_1[a_1',a_1'+d_0\log n-1]\in
D_1$ and

\qquad $X_2'=S_2[b_1',b_1'+d_0\log n-1]\in D_2$ such that $a_1'$ is
the largest and

\qquad $\diff(X_1',X_2')\le \omega_{\algtype}$.

\qquad Find two  subsequences $Y_1=S_1[f_1,f_1+d_0\log n-1]\in D_1$
and

\qquad $Y_2=S_2[e_1,e_1+d_0\log n-1]\in D_2$ such that $e_1$ is the
least and

\qquad $\diff(Y_1,Y_2)\le \omega_{\algtype}$.

\qquad Find two subsequences $Y_1'=S_1[f_1',f_1'+d_0\log n-1]\in
D_1$ and

\qquad $Y_2'=S_2[e_1',e_1'+d_0\log n-1]\in D_2$ such that $e_1'$ is
the largest and

\qquad $\diff(Y_1',Y_2')\le \omega_{\algtype}$.

\qquad Return $(a,a',e_1,e_1')$.

{\bf End of Collision-Detection}

Function Point-Selection$(S_1,S_2,L)$ will be defined differently in
three different algorithms. It selects some positions from each
interval of length $L$ in both $S_1$ and $S_2$.

%It can be considered as a virtual function here.

\vskip 10pt {\bf Point-Selection$(S,L, I)$}

{\bf Input:} a pair of $\Theta(n,G,\alpha)$-sequences $S$,  a size
parameter $L$ of partition, and an interval of positions $I$ in $S$.

{\bf Output:} a set $U$ of positions from $S$ respectively.

 {\bf Steps:}

Let $U=\emptyset$.

If \algtype=\sublinear~or~\randomized

\qquad If $(L\ge \thresholdL)$

\qquad\qquad For each interval $I'$ in $I$, partition $I'$  into
intervals of size $L$.

\qquad\qquad  Sample $M(L)$ random positions at every

\qquad\qquad\qquad\qquad interval of size $L$ derived in the above
partition, and put them into $U$.

\qquad Else

\qquad\qquad Put every position  of $I$ into $U_1$.

If \algtype=\deterministic

\qquad\qquad Put every position  of $I$ into $U$.

Return $U$.

{\bf End of Point-Selection}

%\vskip 10pt {\bf Point-Selection$(S_1,S_2,L)$}

%It selects some points from each interval of length $L$ in both
%$S_1$ and $S_2$. It is based on the \algtype.

%{\bf End of Point-Selection}

%\begin{definition}\label{L0-def}
%For an input sequence $S$ with known rough left boundary
%$\roughleft_{S}$ and right boundary $\roughright_{S}$, define
%$L_0(S,\roughleft_{S},\roughright_{S})=\min(n^{2\over 5},
% (\roughright_{S}-\roughleft_{S})/2)$.
%\end{definition}

\vskip 10pt {\bf Improve-Boundaries$(S_1,a_l, a_r, S_2, f_l,
f_r,L)$}

{\bf Input:} a $\Theta(n,G,\alpha)$-sequence $S_1$ with rough left
and right boundaries $a_l$ and $a_r$, a
$\Theta(n,G,\alpha)$-sequences $S_2$ with rough left and right
boundaries $f_l$ and $f_r$, and the rough distance $L$ to the
nearest motif boundary from those rough boundaries.

{\bf Output:} improved rough left and right boundaries for both
$S_1$ and $S_2$.

 {\bf Steps:}

%\qquad Let $L_0(S_1)=L_0(S_1,\roughleft_{S_1},\roughright_{S_1})$
%and

%\qquad $L_0(S_2)=L_0(S_2,\roughleft_{S_2},\roughright_{S_2})$.

\qquad Find two subsequences $X_1=S_1[a_1,a_1+d_0\log n-1]$ and
$X_2=S_2[b_2,b_2+d_0\log n-1]$

\qquad with $a_1\in [a_l-L, a_l+L]$ and $b_2\in [f_l-L, f_l+L]$ such
that $\diff(X_1,X_2)\le \beta$ and $a_1$ is

\qquad the least.

\qquad Find two subsequences $X_1'=S_1[a_1',a_1'+d_0\log n-1]$ and
$X_2'=S_2[b_2',b_2'+d_0\log n-1]$

\qquad with $a_1'\in [a_r-L, a_r+L]$ and $b_2\in [f_r-L, f_r+L]$
such that  $\diff(X_1',X_2')\le \beta$ and  $a_1'$ is

\qquad the largest.

\qquad Find two subsequences $Y_1=S_1[e_1,e_1+d_0\log n-1]$ and
$Y_2=S_2[f_2,f_2+d_0\log n-1]$

\qquad with $e_1\in [a_l-L, a_l+L]$ and $f_2\in [f_l-L, f_l+L]$ such
that  $\diff(Y_1,Y_2)\le \beta$ and $f_2$ is

\qquad the least.

\qquad Find two subsequences $Y_1'=S_1[e_1',e_1'+d_0\log n-1]$ and
$Y_2'=S_2[f_2',f_2'+d_0\log n-1]$

\qquad with $e_1'\in [a_r-L, a_r+L]$ and $f_2'\in [f_r-L, f_r+L]$
such that  $\diff(Y_1',Y_2')\le \beta$ and $f_2'$ is

\qquad  the largest.

\qquad Return $(a_1,a_1',f_2,f_2')$.

{\bf End of Improve-Boundaries}

\vskip 10pt {\bf Initial-Boundaries$(S_1,S_2)$}

{\bf Input:} a pair of $\Theta(n,G,\alpha)$-sequences $S_1$ and
$S_2$

{\bf Output:} rough left boundary $\roughleft_{S_1}$ of $S_1$, right
boundary $\roughright_{S_1}$ of $S_1$, rough left boundary
$\roughleft_{S_2}$ of $S_2$, and right boundary $\roughright_{S_2}$
of $S_2$.

 {\bf Steps:}

\qquad Let $U_1=U_2=\emptyset$.

\qquad Let $L=n^{2/5}$.

\qquad Repeat

\qquad\qquad Let $U_1=$Point-Selection$(S_1,L, [1,|S_1|])$.

\qquad\qquad Let $U_2=$Point-Selection$(S_2,L, [1,|S_2|])$.

\qquad\qquad Let
$(L_{S_1},R_{S_1},L_{S_2},R_{S_2})=$Collision-Detection$(S_1, U_1,
S_2, U_2)$.

\qquad\qquad If ($L_{S_1}\not=\emptyset$ and
$R_{S_1}\not=\emptyset$)

\qquad\qquad Then Goto H.

\qquad\qquad Else $L=L/2$.

\qquad Until $(L< {1\over 2}\thresholdL)$

\qquad H: Return Improve-Boundaries$(S_1,L_{S_l}, R_{S_1}, S_2,
L_{S_2}, R_{S_2}, 2L)$.

%\qquad Until $L< d_0\log n$.

{\bf End of Initial-Boundaries}

\vskip 10pt

{\bf Motif-Length-And-Boundaries($Z_1$)}

{\bf Input:} $Z_1=\{S_{1}',\cdots,S_{2k_1}'\}$ is a set of
independent $\Theta(n,G,\alpha)$ sequences.

{\bf Steps:}

For $i=1$ to $k_1$

\qquad let
$(\roughleft_{S_{2i-1}'},\roughright_{S_{2i}'})$=Initial-Boundaries$(S_{2i-1}',S_{2i}')$.

Let $L_1$ be the median of
$\cup_{i=1}^{k_1}\{(\roughright_{S_{2i-1}'}-\roughleft_{S_{2i-1}'})\}$.

Return $L_1$.

 {\bf End of Motif-Length-And-Boundaries}

\subsubsection{\phasetwo~ of \algma}

After a set of motif candidates $W$ is produced from \phaseone~ of
algorithm \algmnam, we use this set to match with another set of
input sequences to recover the hidden motif by voting.

\vskip 10pt

{\bf Match$(G_l,G_r, S_i)$}

{\bf Input:} a motif left part $G_l$  (which can be derived from the
rough left boundary of an input sequence $S$), a motif right part
$G_r$, a sequence $S_i''$ from the group $Z_2$, with known rough
left and right boundaries.

{\bf Output:} either a rough motif region of $S_i''$, or an empty
sequence which means the failure in extracting the motif region
$\aleph(S_i'')$ of $S_i''$.

{\bf Steps:}

\qquad Find a position $a$ in $S_i''$ with $\roughleft_{S_i''}\le
a\le \roughleft_{S_i''}+(v+u_2)$.
%, L_0(S_i'',\roughleft_{S_i''}, \roughright_{S_i''})$

\qquad such that $G_l$ and $S_i''[a,a+|G_l|-1]$ are left matched
(see Definition~\ref{match-def}).

\qquad Find a position $b$ in $S_i''$ with
$\roughright_{S_i''}-(v+u_2),
%L_0(S_i'',\roughleft_{S_i''},
\roughright_{S_i''})\le b\le \roughright_{S_i''}$

\qquad such that $G_r$ and $S_i''[b-|G_r|+1,b]$ are right matched
(see Definition~\ref{match-def}).

\qquad If both $a$ and $b$ are found

\qquad Then output $S_i''[a,b]$

\qquad Else output $\emptyset$ (empty string).

{\bf End of Match}

\vskip 10pt

{\bf Extract$(G_l, G_r, Z_2$):}

Input $Z_2=\{S_1'',S_2'',\cdots, S_{k_2}''\}$ and their rough left
boundaries and rough right boundaries.
%in the begining of Section~\ref{algorithm-sec}.

{\bf Steps:}

\qquad For each $S_i''$ with $i=1,2,\cdots, k_2$,

\qquad\qquad let $G_i''=\match(G_l, G_r, S_i'')$.

\qquad  Return $(G_1'', G_2'',\cdots, G_{k_2}'')$.

{\bf End of Extract}

\vskip 10pt

The following is \phasetwo~ of algorithm \algmnam. It  extracts the
motif regions of another set $Z_2$ of input sequences.

{\bf \phasetwo($S', Z_2$):}

Input $S'$ is an input sequence with known $\roughleft_{S'}$ and
$\roughright_{S'}$ for its rough left and right boundaries
respectively, and $Z_2=\{S_1'',\cdots, S_{k_2}''\}$ is a set of
input sequences.
%in the begining of Section~\ref{algorithm-sec}.

{\bf Steps:}

\qquad For each subsequence $G_l=S'[a, a+d_0\log n-1]$ with $a\in
[\roughleft_{S'}, \roughleft_{S'}+(v+u_1)]$

\qquad and $G_r=S'[b-d_0\log n+1,b]$ with $b\in
[\roughright_{S'}-(v+u_1), \roughright_{S'}]$

\qquad\qquad let $(G_1'', G_2'',\cdots, G_{k_2}'')$ be the output
from Extract$(G_l, G_r, Z_2$).

\qquad\qquad  If  the number of empty sequences in $G_1'',\cdots,
G_{k_2}''$ is at most $(Q_0+(R+2\epsilon))k_2$

\qquad\qquad Then return $(G_1'', G_2'',\cdots, G_{k_2}'')$.

\qquad Return $\emptyset$ (empty set).

{\bf End of \phasetwo~}

\subsubsection{\phasethree}

The function Vote$(G_1'', G_2'',\cdots, G_{k_2}'')$ is to generate
another sequence $G'$ by voting, where $G'[i]$ is the most frequent
character among $G_1''[i], G_2''[i],\cdots, G_{k_2}''[i]$.

 \vskip 10pt

{\bf \phasethree$(G_1'', G_2'',\cdots, G_{k_2}'')$}

{\bf Input:} $\Theta(n,G,\alpha)$ sequences $G_1'', G_2'',\cdots,
G_{k_2}''$ of the same length $m$.

{\bf Output:} a sequence $G'$, which is derived by voting on every
position of the input sequences.

 {\bf Steps:}

%\qquad Let $m$ be the most frequent length among $|G_1''|,\cdots, |G_{k_2}''|$.

\qquad For each $j=1,\cdots, m$

\qquad\qquad let $a_j$ be the most frequent character among
$G_1''[j],\cdots, G_{k_2}''[j]$.

\qquad Return $G'=a_1\cdots a_{m}$.

{\bf End of Vote}

\subsubsection{Entire \algma}

 The entire algorithm is described below. We maintain the size of $Z_1$ and $Z_2$ to be
roughly equal, which implies
\begin{eqnarray}
|Z_1|=\Theta(|Z_2|) \label{Z1-Z2-eqn}
\end{eqnarray}

{\bf \algma(Z)}

Input: $Z=Z_1\cup Z_2$, where $Z_1=\{S_1',\cdots, S_{2k_1}'\}$ and
$Z_2=\{S_1'',\cdots, S_{k_2}''\}$ are two sets of input sequences.

{\bf Steps:}

{\bf Preprocessing Part:}

For each $S\in Z_1\cup Z_2$, let $\roughleft_S=\roughright_S=0$ (the
two boundaries are unknown).

$l_{motif}=$MotifLengthAndBoundaries($Z_1$).

Let $L=l_{motif}/4$.

For $i=1$ to $k_1$,

\qquad let $U_{S_{2i-1}'}=$Point-Selection$(S_{2i-1}',L,
[\roughleft_{S_{2i-1}'}-2L, \roughleft_{S_{2i-1}'}+2L])\cup$

\qquad \qquad \qquad\hskip 13pt Point-Selection$(S_{2i-1}',L,
[\roughright_{S_{2i-1}'}-2L, \roughright_{S_{2i-1}'}+2L])$.

For $j=1$ to $k_2$

\qquad let $U_{S_j''}=$Point-Selection$(S_j'',L,[1,|S_j''|])$.

For $i=1$ to $k_1$

%\qquad All-Rough-Boundaries($S_{2i-1}', Z_2,   l_{motif}$).

\qquad For each $S_j''\in Z_2$

\qquad\qquad Let $(L_{S_{2i-1}'}, R_{S_{2i-1}'},
L_{S_j''},R_{S_j''})=$Collision-Detection$(S_{2i-1}', U_{S_{2i-1}'},
S_j'', U_{S_j''})$.

\qquad\qquad Let $(L_{S_{2i-1}'}, R_{S_{2i-1}'},
\roughleft_{S_j''},\roughright_{S_j''})$=

\qquad\qquad\qquad Improve-Boundaries$(S_{2i-1}',L_{S_{2i-1}'},
R_{S_{2i-1}'}, S_j'', L_{S_j''}, R_{S_j''}, 2L)$.

\qquad Let $(G_1'', G_2'',\cdots, G_{k_2}'')$ be the output from
\phasetwo($S_{2i-1}', Z_2$).

\qquad If $(G_1'', G_2'',\cdots, G_{k_2}'')$ is not empty

\qquad Then go to Voting Part.

{\bf Voting Part:}

\qquad Return \phasethree($G_1'', G_2'',\cdots, G_{k_2}''$).

{\bf End of \algma}

%\vskip 10pt

%The preprocessing part is the combination of \phaseone~and
%\phasetwo~for deriving $(G_1'', G_2'',\cdots, G_{k_2}'')$ to be used
%for voting to recover the motif. The second part is voting that is
%described by the executing \phasethree($G_1'', G_2'',\cdots,
%G_{k_2}''$).

%P.S. For right match, we just use the existing left
%match(reverse(S), reverse(S')) to deal with right match(S,S').

\section{Analysis of Algorithm}\label{analysis-sec}
 The correctness
of the algorithm will be proved via a series of Lemmas in
Sections~\ref{sec-analysis} and~\ref{sec-analysis2}.
Section~\ref{sec-analysis} is for \phaseone~ and
Section~\ref{sec-analysis2} is for \phasetwo. Furthermore,
Section~\ref{sec-analysis2} gives some lemma for the two randomized
algorithms and Section~\ref{deterministic-sec} gives the proof for
the deterministic algorithm.

\subsection{Review of Some Classical Results in Probability}

Some well known results in classical probability theory are listed.
The readers can skip this section if they understand them well. The
inclusion of these results make the paper self-contained.
\begin{itemize}
\item
 For a list of events $A_1,\cdots, A_m$,
$\prob[A_1\cup A_2\cup \cdots \cup A_m]\le
\prob[A_1]+\prob[A_2]+\cdots+\prob[A_m]$.
\item
For two independent events $A$ and $B$, $\prob[A\cap
B]=\prob[A]\prob[B]$.
\item
For a random variable $Y$, $\prob[Y\ge t]\le {E[Y]\over t}$ for all
positive real number $t$. This is called Markov inequality.
\end{itemize}

The analysis of our algorithm employs the Chernoff bound
\cite{MotwaniRaghavan00} and Corollary~\ref{chernoff-lemma-a} below,
which can be derived from it
%Theorem~\ref{chernoff-theorem}
(see~\cite{LiMaWang99}).

\begin{theorem}[\cite{MotwaniRaghavan00}]\label{chernoff-theorem}
Let $X_1,\cdots, X_n$ be $n$ independent random $0$-$1$ variables,
where $X_i$ takes $1$ with probability $p_i$. Let $X=\sum_{i=1}^n
X_i$, and $\mu=E[X]$. Then for any $\delta>0$,
\begin{enumerate}
\item $\Pr(X<(1-\delta)\mu)<e^{-{1\over 2}\mu\delta^2}$, and
\item
$\Pr(X>(1+\delta)\mu)<\left[{e^{\delta}\over
(1+\delta)^{(1+\delta)}}\right]^{\mu}$.
\end{enumerate}
\end{theorem}

We follow the proof of Theorem~\ref{chernoff-theorem} to make the
following version of Chernoff bound so that it can be used in our
algorithm analysis. %Its proof makes our entire paper self-contained.

\begin{theorem}\label{ourchernoff-theorem}
Let $X_1,\cdots, X_n$ be $n$ independent random $0$-$1$ variables,
where $X_i$ takes $1$ with probability at most $p$. Let
$X=\sum_{i=1}^n X_i$. Then for any $\delta>0$,
$\Pr(X>(1+\delta)pn)<\left[{e^{\delta}\over
(1+\delta)^{(1+\delta)}}\right]^{pn}$.
\end{theorem}

\begin{proof} Let $y$ be an arbitrary positive real number. By the
definition of expectation, we have
$E(e^{yX_i})=\Pr(X_i=1)e^y+\Pr(X_i=0)$. Since the function
$f(x)=xe^y+(1-x)$ is increasing for all $y>0$ and $\Pr(X_i=1)\le p$,
we have $E(e^{yX_i})\le pe^y+(1-p)$. We have the following
inequalities:
\begin{eqnarray}
\Pr(X>(1+\delta)pn)&<&{E(e^{yX})\over e^{y(1+\delta)pn}}\label{markov-inequ}\\
&\le&{\prod_{i=1}^nE(e^{yX_i})\over e^{y(1+\delta)pn}}\label{independent1}\\
&=&{\prod_{i=1}^n(pe^y+1-p)\over e^{y(1+\delta)pn}}\label{independent2}\\
&=&{\prod_{i=1}^n(1+p(e^y-1))\over e^{y(1+\delta)pn}}\\
&\le&{\prod_{i=1}^ne^{p(e^y-1)}\over e^{y(1+\delta)pn}}\\
&=&{e^{(e^y-1)pn}\over e^{y(1+\delta)pn}}\\
&=&({e^{(e^y-1)}\over e^{y(1+\delta)}})^{pn}\label{final1}.
\end{eqnarray}
The inequality (\ref{markov-inequ}) is based on Markov inequality.
The transition from (\ref{independent1}) to (\ref{independent2}) is
due to the independence of those variables $X_1,\cdots, X_n$.

Since $({e^{(e^y-1)}\over e^{y(1+\delta)}})$ is minimal at
$y=\ln(1+\delta)$, we have
$\Pr(X>(1+\delta)pn)<\left[{e^{\delta}\over
(1+\delta)^{(1+\delta)}}\right]^{pn}$.
\end{proof}

%Define $g_1(\delta)=e^{-{1\over 2}\delta^2}$ and
Define $g(\delta)={e^{\delta}\over (1+\delta)^{(1+\delta)}}$.
%Define $g(\delta)=\max(g_1(\delta),g_2(\delta))$.
We note that $g(\delta)$ is always strictly less than $1$ for all
$\delta>0$, and $g(\delta)$ is fixed if $\delta$ is a constant. This
can be verified by checking that the function $f(x)=\ln {e^x\over
(1+x)^{(1+x)}}=x-(1+x)\ln (1+x)$ is decreasing and $f(0)=0$. This is
because $f'(x)=-\ln (1+x)$, which is less than $0$ for all $x>0$.

\begin{theorem}\label{ourchernoff-theorem}
Let $X_1,\cdots, X_n$ be $n$ independent random $0$-$1$ variables,
where $X_i$ takes $1$ with probability at most $p$. Let
$X=\sum_{i=1}^n X_i$. Then for any $\delta>0$,
$\Pr(X>(1-\delta)pn)<e^{-{1\over 2}^{pn}\delta^2}$.
\end{theorem}

\begin{proof}
$\prob[X<(1-\delta)pn]=\prob[-X>-(1-\delta)pn]=\prob[e^{-yX}>e^{-y(1-\delta)pn}]$
for each real number $y$. Applying Markov inequality, we have
\begin{eqnarray}
\prob[X<(1-\delta)pn]&<&{\prod_{i=1}^nE(e^{-yX_i}]\over
e^{-y(1-\delta)np}}\label{chernoff-1}\\
&<&{e^{(e^{-y}-1)np}\over e^{-y(1-\delta)np}}\label{chernoff-2}\\
&<&\left( {e^{-\delta}\over
(1-\delta)^{1-\delta}}\right)^{pn}\label{chernoff-3}\\
&<&e^{-{1\over 2}^{pn}\delta^2}.\label{chernoff-4}
\end{eqnarray}
The transition from~(\ref{chernoff-2}) to (\ref{chernoff-2}) is to
let $t=\ln {1\over 1-\delta}$. The transition
from~(\ref{chernoff-3}) to (\ref{chernoff-4}) follows from the fact
$(1-\delta)^{1-\delta}>e^{-\delta+\delta^2/2}$.
\end{proof}

\begin{corollary}[\cite{LiMaWang99}]\label{chernoff-lemma-a}
Let $X_1,\cdots, X_n$ be $n$ independent random $0$-$1$ variables
and $X=\sum_{i=1}^n X_i$.

 i. If $X_i$ takes $1$ with probability at most $p$, then for any ${1\over 3}>\epsilon>0$, $\Pr(X>pn+\epsilon
n)<e^{-{1\over 3}n\epsilon^2}$.
%\end{corollary}

%\begin{corollary}\label{chernoff-lemma-b}
ii. If  $X_i$ takes $1$ with probability at least $p$, then for any
$\epsilon>0$, $\Pr(X<pn-\epsilon n)<e^{-{1\over 2}n\epsilon^2}$.
\end{corollary}

%\begin{lemma}[\cite{LiMaWang02}]\label{chernoff-lemma}
%Let $X_1,\cdots, X_n$ be $n$ independent random $0,1$ variables,
%where $X_i$ takes $1$ with probability $p$. Let $X=\sum_{i=1}^n
%X_i$. Then for any ${1\over 3}>\epsilon>0$, (1) $Pr(X<pn-\epsilon
%n)<e^{-{1\over 2}n\epsilon^2}$, and (2)$Pr(X>pn+\epsilon
%n)<e^{-{1\over 3}n\epsilon^2}$.
%\end{lemma}

\begin{proof}
For $X=\sum_{i=1}^n$, $\mu=E(X)=\sum_{i=1}^nE(X_i)=pn$. Let
$\delta={\epsilon\over p}$. (1) follows from
Theorem~\ref{chernoff-theorem}.  By Taylor theorem,
$\ln(1+\epsilon)\ge \epsilon-{\epsilon^2\over 2}$. We have that
$(1+{1\over \epsilon})\ln(1+\epsilon)\ge(1+{1\over
\epsilon})(\epsilon-{\epsilon^2\over 2})=1+{\epsilon\over
2}-{\epsilon^2\over 2}>1+{\epsilon\over 3}$. Thus, $ {e\over
(1+\epsilon)^{(1+{1\over \epsilon})}}<e^{-{\epsilon\over 3}}$. Since
$pn+\epsilon n=(1+\delta)\mu$ and the function $(1+y)^{1\over y}$ is
increasing for $y>0$, $Pr(X>pn+\epsilon
n)=Pr(X>(1+\delta)\mu)<\left[ {e^{\epsilon\over p}\over
(1+{\epsilon\over p})^{(1+{\epsilon\over p})}}\right]^{pn}=\left[
{e\over (1+{\epsilon\over p})^{(1+{p\over
\epsilon})}}\right]^{\epsilon n}\le \left[ {e\over
(1+\epsilon)^{(1+{1\over \epsilon})}}\right]^{\epsilon n}\le
e^{-{\epsilon^2n\over 3}}.$ Thus (ii) is proved.
\end{proof}

\subsection{Analysis of \phaseone~ of  \algma}\label{sec-analysis}

%\begin{definition}\scrod
%\begin{itemize}
%\item
%Assume that \algtype~is either \sublinear~or \randomized. Let $A$ be
%a set of $L$ elements, $B$ be an another set $L$ elements, and $C$
%be a subset in $A\cap B$ with $|C|\ge {L\over 2}$. Let $U_1$ be a
%subset of $M(L)$ random elements of $A$, and $U_2$ be a subset of
%$M(L)$ random elements of $B$. Define $q_{\rm
%selection}(L)=\prob(U_1\cap U_2\cap C=\emptyset)$.
%\item
%Assume that \algtype~is \deterministic. Let $q_{\rm
%selection}(L)=0$.
%\end{itemize}
%\end{definition}

%\begin{lemma}\label{collision-lemma}
%Assume that CollisionDetection($S_1, U_1, S_2, U_2$) returns
% the result in time $T(n,||U_1||+||U_2||)$ time. Then
%InitialBoundary$(S_1,S_2)$ runs in time $O(\sum_{i=1}^{i_0}T(n,
%{n\over 2^in^{2/5}})(\log n)^2)$, where $i_0$ is the least $j$ such
%that ${n\over 2^jn^{2/5}}\le |G|$.
%\end{lemma}

%\begin{proof} It follows from the implementation of Initial-Boundaries() and the
%way it calls the function Collision-Detection().
%\end{proof}

%\begin{lemma}\label{no-collision-lemma}
%with probability at most ${1\over 2^xn^{3}}$, there is a pair of
%intervals such that one of the two intervals is has distance more
%than $L+???\log$ with the left boundary,   and there is a collision
%between two points of them.
%\end{lemma}

%\begin{proof}
%If a point has at least $???\log$ distance to the left boundary, the
%probability is at most $e^{\epsilon^2 ???\log n\over 3}{1\over
%2^{???x}n^{100???}}$ that a collision will happen by the Chernoff
%bound. This is because a random character at each position.
%\end{proof}

Lemma~\ref{base-probability-lemma} shows that with only small
probability, a sequence can match a random sequence. It will be used
to prove that when two substrings in two different
$\Theta(n,G,\alpha)$-sequences are similar, they are unlikely not to
coincide with  the motif regions in the two
$\Theta(n,G,\alpha)$-sequences, respectively.

\begin{lemma}\label{base-probability-lemma}
Assume that $X_1$ and $X_2$ are two independent sequences of the
same length and that every character of $X_2$ is a random character
from $\Sigma$. Then
\begin{enumerate}
\item
 if $1\le |X_1|=|X_2|< v$, then the probability that $X_1$ and
$X_2$ are matched is $\le {1\over t^{|X_1|}}$ ($t=||\Sigma||$); and
\item
 the probability for $\diff(X_1, X_2)\le \beta$ is at most $ e^{-{\epsilon^2|X_1|\over 3}}$.
\end{enumerate}
\end{lemma}

\begin{proof} The two statements are proved as follows.

Statement i: For every character $X_2[j]$ with $1\le j<v$, the
probability is ${1\over t}$ that $X_2[j]=X_1[j]$.

\hskip 35pt Statement ii: For every character $X_2[j]$ with $1\le
j\le |X_2|$, the probability is ${1\over t}$ for $X_2[j]$ to equal
$X_1[j]$. If $\diff(X_1, X_2)\le \beta$, the two sequences $X_1$ and
$X_2$ are identical at least $(1-\beta)|X_1|$ positions, but the
expected number of positions where the two sequences are identical
is ${1\over t}|X_1|$. The probability for $\diff(X_1, X_2)\le \beta$
is at most $e^{-{(1-\beta-{1\over t})^2\over 3}|X_1|}\le
e^{-{\epsilon^2\over 3}|X_1|}$ by Corollary~\ref{chernoff-lemma-a},
and Definitions~\ref{param-def} and~\ref{match-def}.
%Section~\ref{parameters-subsect}).
\end{proof}

Lemma~\ref{diff-motif-lemma} shows that with small probability, an
input $\Theta_{\alpha}(n,G)$ sequence contains motif region that has
many mutations.

\begin{lemma}\label{diff-motif-lemma}
With probability at most ${c^y\over 1-c}$,  a $\Theta_{\alpha}(n,G)$
sequence $S$ changes more than ${\beta\over 2}t$ characters in its
first left $t$ motif region $\aleph(S)$ for some $t$ with $y\le t\le
|G|$, where $c=e^{-{\epsilon^2\over 3}}$.
\end{lemma}

\begin{proof}
Every character in the $\aleph(S)$ region has probability at most
$\alpha$ to mutate. We know that $|\aleph(S)|=|G|\ge d$.  By
Corollary~\ref{chernoff-lemma-a}, with probability at most
$e^{-{\epsilon^2\over 3}t}$, a sequence $S$ in $Z_1$ has more than
$(\alpha+\epsilon)t$ mutations (recall the setting for $\beta$ at
Definition~\ref{match-def}) among the first left $t$ characters. The
total is $\sum_{t=y}^{\infty} e^{-{\epsilon^2\over 3}t}={c^y\over
1-c}$.
\end{proof}

Lemma~\ref{v+u-lemma2} shows that Improve-Boundaries() has good
chance to improve the accuracy of rough motif boundaries.

\begin{lemma}\label{v+u-lemma2}
Assume that $\Theta_{\alpha}(n,G)$ sequence $S_i$ has $L_{S_i}\in
[\LB(S_i)-L,\LB(S_i)+L]$
 and $R_{S_i}\in [\RB(S_i)-L,\RB(S_i)+L]$ for $i=1,2$. Then
 for $(\roughleft_{S_1},\roughright_{S_1},\roughleft_{S_2},$ $\roughright_{S_2})$=Improve-Boundaries$(S_1,L_{S_1},R_{S_1},S_2,L_{S_2},R_{S_2},L)$,
 we  have the following two facts:
\begin{enumerate}
\item
 With probablity at most
${2c^{v}\over 1-c}+{2(v+u)c^{v+u}\over (1-c)^2}+{1\over 5\cdot 2^x
n}$, $\roughleft_{S_i}$ is not in $[\LB(S_i)-(v+u),\LB(S_i)]$ for
$i=1,2$.
\item
 With probablity at most
${2c^{v}\over 1-c}+{2(v+u)c^{v+u}\over (1-c)^2}+{1\over 5\cdot 2^x
n}$, $\roughright_{S_i}$ is not in $[\RB(S_i),\RB(S_i)+(v+u)]$ for
$i=1,2$.
\item
Improve-Boundaries$(S_1,L_{S_1},R_{S_1},S_2,L_{S_2},R_{S_2},L)$ runs
in $O(L^2\log n)$ time.
\end{enumerate}
\end{lemma}

\begin{proof}
We need a bound for the following inequality:
\begin{eqnarray}
 \sum_{i=j}^{\infty}ia^i< {ja^j\over (1-a)^2}.\label{sum-eqn}
\end{eqnarray}
 Let $f(x)=\sum_{i=j}^{\infty}e^{\theta i x}$. Compute the
derivative $f'(x)=\theta \sum_{i=j}^{\infty} ie^{\theta i x}$. We
also have the closed form for the function $f(x)={e^{\theta j
x}\over 1-e^{\theta x}}$, which implies
\begin{eqnarray}
f'(x)&=&{\theta j e^{\theta j x}(1-e^{\theta x})-e^{\theta j
x}(-\theta e^{\theta x})\over (1-e^{\theta x})^2}\\
&=&{\theta j e^{\theta j x}-\theta (j-1)e^{\theta (j+1) x}\over
(1-e^{\theta x})^2}. \label{f'-inequality}
\end{eqnarray}

 Let $\theta=\ln a$ and $x=1$. We have
 $\sum_{i=j}^{\infty}ia^i={ja^j-(j-1)a^{j+1}\over (1-a)^2}<
{ja^j\over (1-a)^2}$.

Statement i. By Lemma~\ref{diff-motif-lemma}, with probability at
most $2{c^{v}\over 1-c}$, one of the left motif first $y$ characters
region of $S_i$ will change ${\beta\over 2}y$ characters. Therefore,
with probability at most $P_1=2{c^{v}\over 1-c}$,
$\roughleft_{S_i}>\LB(S_i)$.

For a pair of positions $p$ in $S_1$ and $q$ in $S_2$, without loss
generality, assume that $p$ has larger distance to the left boundary
$\LB(S_1)$ of $S_1$ than $q$ to the left boundary $\LB(S_2)$ of
$S_2$. Let $v+y$ be the distance from $p$ to the left boundary
$\LB(S_1)$ of $S_1$.

By Lemma~\ref{base-probability-lemma}, the probability is at most
$c^{v+y}$ that there will be a match. There are at most $(v+y)$
cases for $q$. With probability is at most $P_2=2\sum_{y=u}^{\infty}
(v+y)c^{v+y}<{2(v+u)c^{v+u}\over (1-c)^2}$ by inequality
(\ref{sum-eqn}), $\roughleft_{S_1}<LB(S_1)-(v+u)$.

For the cases that one position  is in random region and has
distance more than $d_0\log n$ with the left boundary, the
probability is at most $P_3=n^2c^{d_0\log n}<{1\over 5\cdot 2^x n}$
by inequality~(\ref{d0-sel-eqn}).

Therefore, we have total probability at most $P_1+P_2+P_3$ that
$\roughleft_{S_1}$ is not in $[\LB(S_1)-(v+u),\LB(S_1)]$.

Statement ii. One can also provide a symmetric analogous proof for
this statement.

Statement iii. The computation time easily follows from the
implementation of
Improve-Boundaries$(S_1,L_{S_1},R_{S_1},S_2,L_{S_2},R_{S_2})$.
\end{proof}

%\begin{definition}\scrod
%\begin{itemize}
%\item
%Define $z(n)$ to be the least $L$ used by the algorithm to partition
%the sequence into intervals of size $L$.
%\item
%Define $p_{\rm selection}(n)=\max_{n^{2\over 5}\ge L\ge
%z(n)}\{q_{\rm selection}(L)\}$.
%\end{itemize}
%\end{definition}

%\begin{lemma}\scrod
%\begin{enumerate}
%\item
%With probability at most $(p_{\rm selection}(n)+ {2(v+u)c^{v+u}\over
%(1-c)^2})^{k_1}$, each $\roughleft_{S_{2i-1}}$ ($i=1,\cdots, k_1$)
%returned by function Initial-Boundaries($S_{2i-1},S_{2i}$) is not in
%$[\LB(S_1),\LB(S_1)+(v+u)]$.
%\item
%With probability at most $(p_{\rm selection}(n)+ {2(v+u)c^{v+u}\over
%(1-c)^2})^{k_1}$, each $\roughright_{S_{2i-1}}$ ($i=1,\cdots, k_1$)
%returned by function Initial-Boundaries($S_{2i-1},S_{2i}$) is not in
%$[\RB(S_1)-(v+u),\RB(S_1)]$.
%\end{enumerate}
%\end{lemma}

%\begin{proof}
%It follows from Lemmas~\ref{v+u-lemma2}, \ref{collision-lemma}, and
%\ref{collision-lemma2}.
%\end{proof}

\begin{lemma}\label{initial-boundary-lemma}
Assume that for each $L$ with $0<L\le {|G|\over 2}$, with
probability at most $\varsigma(n)$, $L_{S_i}\not \in [\LB_{S_i}-L,
\LB_{S_i}+L]$ for $i=1,2$, where $(L_{S_1}, R_{S_1}, L_{S_2},
R_{S_2})=$Collision-Detection($S_1, U_1, S_2, U_2)$,
$U_1=$Point-Selection$(S_1,L)$, and $U_2=$Point-Selection$(S_2,L)$.
 Then with probability at most $\varsigma(n) +
{2(v+u_1)c^{v+u_1}\over (1-c)^2}+{c^v\over 1-c}+{1\over 5\cdot 2^x
n}$, Initial-Boundary$(S_1,S_2)$ returns $(L_{S_1}, R_{S_1},
L_{S_2}, R_{S_2})$ with  $L_{S_i}\not\in
[\LB(S_i)-(v+u_1),\LB(S_i)]$
 or $R_{S_i}\not\in [\RB(S_i),\RB(S_i)+(v+u_1))]$ for $i=1,2$;
\end{lemma}

\begin{proof}
It follows from Lemma~\ref{v+u-lemma2}.
\end{proof}

\begin{lemma}\label{select-median-lemma}
Assume that with probability $p<0.5$, each $S_{2i-1}'$ has its rough
boundaries $\roughleft_{S_{2i-1}'}\not\in[\LB(S_{2i-1}')-u,
\LB(S_{2i-1}')]$ or $\roughright_{S_{2i-1}'}\not\in[\RB(S_{2i-1}'),
\RB(S_{2i-1}')+u]$, then with probability at most
$e^{-{(0.5-p-\epsilon)^2k_1/3}}$, $l_{motif}$ is not in $[|G|-2u,
|G|+2u]$, where $l_{motif}$ is selected as median of
$\cup_{i=1}^{k_1}\{(\roughright_{S_{2i-1}'}-\roughleft_{S_{2i-1}'})\}$.
\end{lemma}

\begin{proof}
If both $\roughleft_{S_{2i-1}'}\in[\LB(S_{2i-1}')-u,
\LB(S_{2i-1}')]$ and $\roughright_{S_{2i-1}'}\in[\RB(S_{2i-1}'),
\RB(S_{2i-1}')+u]$, then
$(\roughright_{S_{2i-1}'}-\roughleft_{S_{2i-1}'})$ is in $[|G|-2u,
|G|+2u]$.

If the median of
$\cup_{i=1}^{k_1}\{(\roughright_{S_{2i-1}'}-\roughleft_{S_{2i-1}'})\}$
is not in $[|G|-2u, |G|+2u]$, then there are at least $\floor{k_1}$
$i$s to have $\roughleft_{S_{2i-1}'}\not\in[\LB(S_{2i-1}')-u,
\LB(S_{2i-1}')]$ or $\roughright_{S_{2i-1}'}\not\in[\RB(S_{2i-1}'),
\RB(S_{2i-1}')+u]$.

On the other hand, the probability is at most $p$,
$\roughleft_{S_{2i-1}'}\not\in[\LB(S_{2i-1}')-u, \LB(S_{2i-1}')]$ or
$\roughright_{S_{2i-1}'}\not\in[\RB(S_{2i-1}'), \RB(S_{2i-1}')+u]$.
So,  this lemma follows from Corollary~\ref{chernoff-lemma-a}.
\end{proof}

For a $\Theta(n,G,\alpha)$-sequence $S$, we often obtain its left
rough boundary with $\roughleft_{S}\le\LB(S)$. Some times its
exactly left boundary may be miss in the algorithm.

\begin{definition}\scrod
\begin{itemize}
\item
A $\Theta(n,G,\alpha)$-sequence $S$ misses its left boundary if
$\roughleft_{S}>\LB(S)$.
\item
A $\Theta(n,G,\alpha)$-sequence $S$ misses its right boundary if
 $\roughright_{S}<\RB(S)$.
\end{itemize}
\end{definition}

\begin{definition}\label{stable-def}\scrod
\begin{itemize}
\item
  A $\Theta(n,G,\alpha)$-sequence $S$ contains a {\it left half stable} motif region $\aleph(S)$
if $\diff(G'[1, h],G[1, h])\le {\beta\over 2}$ for all
$h=v,v+1,\cdots, m$, where $G'=\aleph(S)$, $c=e^{-{\epsilon^2\over
3}}$ and $m=|G|$ as defined in Definition~\ref{param-def}
% hahaha
%Sections~\ref{parameters-subsect}
and Section~\ref{notation-sec}, respectively.
\item
 A $\Theta(n,G,\alpha)$-sequence $S$ contains a {\it right half stable} motif region $\aleph(S)$
if $\diff(G'[m-h, m], G[m-h, m])\le {\beta\over 2}$ for
$h=v-1,v+1,\cdots, m-1$, where $G'=\aleph(S)$ and $m=|G|$.
\item
  A $\Theta(n,G,\alpha)$-sequence $S$ contains a {\it stable} motif region $\aleph(S)$
satisfying  the following conditions: (1) $G'[i]=G[i]$ for
$i=1,\cdots, v-1$; (2) $G'[m-i+1]=G[m-i+1]$ for $i=1,\cdots, v-1$;
(3) $S$ motif region is both left and right half stable, where
$G'=\aleph(S)$ and $m=|G|$.
\end{itemize}
\end{definition}

\begin{lemma}\label{left-boundary-detect-lemma}
 Assume that
\begin{itemize}
\item
  $l_{motif}\in [|G|-2(v+u_1), |G|+2(v+u_1)]$;
\item
 $S$ contains a both left half and right half stable motif region and $\roughleft_{S}\in [\LB(S)-(v+u_1),\LB(S)]$
 and $\roughright_{S}\in [\RB(S),\RB(S)+(v+u_1)]$ (see Definition~\ref{param-def} for $u_1$ and
 $v$); and
\item
for each $L$ with $(v+u_1)<L\le {|G|\over 2}$, if $S_1$ has
$\roughleft_{S_1}\not \in [\LB_{S_1}-L, \LB_{S_1}+L]$ and
$\roughright_{S_1}\not \in [\RB_{S_1}-L, \RB_{S_1}+L]$, then
 with probability at most
$\varsigma(n)$, $L_{S_i''}\not \in [\LB_{S_i''}-2L, \LB_{S_i''}+2L]$
for $i=1,2$, where $(L_{S_1}, R_{S_1}, L_{S_i''},
R_{S_i''})=$Collision-Detection($S_1, U_1, S_i'', U_2)$,
$U_1=$Point-Selection$(S_1,L,
[\roughleft_{S_1}-2L,\roughleft_{S_1}+2L])\cup$
Point-Selection$(S_1,L,
[\roughright_{S_1}-2L,\roughright_{S_1}+2L])$, and
$U_2=$Point-Selection$(S_i'',$ $L, [1,|S_i''|])$.
\item
The rough boundaries for all sequences $S_i''\in Z_2$ are computed
via $(L_{S}, R_{S}, L_{S_i''},R_{S_i''})=$Collision-Detection$(S,
U_{S}, S_i'', U_{S_i''})$, and $(L_{S}, R_{S},
\roughleft_{S_i''},\roughright_{S_i''})$=Improve-Boundaries$(S,$
$L_{S}, R_{S}, S_i'', L_{S_i''}, R_{S_i''}, 2L)$.
\end{itemize}

 Then with probability at most $e^{-{\epsilon^2k_2\over 3}}$, there are
more than $(2(\varsigma(n)+ (v+u_1){c^{v+u}\over 1-c}+{c^{v}\over
1-c})+\epsilon)k_2$ sequences $S_i''$ in $\{S_1'',\cdots,
S_{k_2}''\}$ with $\roughleft(S_i'')\not\in [\LB(S_i'')-(v+u),
\LB(S_i'')]$ or $\roughright(S_i'')\not\in [\RB(S_i''),
\RB(S_i'')+(v+u)]$.
\end{lemma}

\begin{proof}
%By Lemma~\ref{collision-lemma}, with probability at most
%$P_{1,l}={1\over 2^{x}n^{3}}$, a collision does not happen near the
%true left boundary.  By Lemma~\ref{collision-lemma}, with
%probability at most $P_{1,r}={1\over 2^{x}n^{3}}$, a collision does
%not happen near the true right boundary. Therefore, with probability
%at most $P_1=P_{1,l}+P_{1,r}$, a collision does not happen near ear
%the true left boundary or the true right boundary
According to the condition of this lemma, with probability at most
$P_1=\varsigma(n)$, $L_{S_i''}\not\in [\LB_{S_i''}-2L,
\LB_{S_i''}+2L]$, where $(L_{S}, R_{S}, L_{S_i''},
R_{S_i''})=$Collision-Detection($S, U_1, S_i'', U_2)$ and $(U_1,
U_2)=$Point-Selection$(S,S_i'',L)$.

For a fix pattern from $S$, by Lemma~\ref{base-probability-lemma},
with probability at most $\sum_{y=v+u}^{\infty}c^y={c^{v+u}\over
1-c}$, it has distance more than $v+u$ to the true left boundary. As
we need to deal with $v+u_1$ possible patterns from $S$, with
probability at most $P_{2,l}= (v+u_1){c^{v+u}\over 1-c}$,
$\roughleft_{S_i''}< \LB(S_i'')-(v+u)$.

Similarly, with probability at most $P_{2,r}= (v+u_1){c^{v+u}\over
1-c}$, $\roughright_{S_i''}< \RB(S_i'')+(v+u)$. Let
$P_2=P_{2,l}+P_{2,r}$.

With probability at most $P_{3,l}={c^{v}\over 1-c}$,  $S_i''$ does
not contain a left half stable motif region by
Lemma~\ref{diff-motif-lemma}. Similarly, with probability at most
$P_{3,r}={c^{v}\over 1-c}$,  $S_i''$ does not contain a right half
stable motif region. Let $P_3=P_{3,l}+P_{3,r}$.

Although $S$ is involved to search the left boundary with all other
sequences. The non-missing condition is to let each sequence do not
change too many characters in the motif region. Therefore, this is
an independent event for each sequence. It is safe to use Chernoff
bound to deal with it.

With probability at most $P=e^{-{\epsilon^2k_2\over 3}}$, the are
more than $(P_1+P_2+P_3+\epsilon)k_2$ sequences $S_i''$ in
$\{S_1'',\cdots, S_{k_2}''\}$ with $\roughleft(S_i'')\not\in
[\LB(S_i'')-(v+u), \LB(S_i'')]$ or $\roughright(S_i'')\not\in
[\RB(S_i''), \RB(S_i'')+(v+u)]$.

\end{proof}

\subsection{Analysis of \phasetwo~and \phasethree~of  \algma}\label{sec-analysis2}

\begin{figure}{
\vskip 100pt
  {\begin{picture}(5.0,5.0)

      \put(0.0, 75.0){{$G''$} }

      \put(190.0, 20.0){{$\aleph(S)$} }

      \put(180.0, 23.0){\begin{picture}(0.0,0.0)
                           \thinlines{\vector(-1,0){130.0}}
                         \end{picture}
                        }

      \put(215.0, 23.0){\begin{picture}(0.0,0.0)
                           \thinlines{\vector(1,0){150.0}}
                         \end{picture}
                        }

    \put(380.0, 30.0){ $w'$ }

    \put(25.0, 30.0){ $w$ }

      \put(0.0, 33.0){$M$}

      \put(20.0, 40.0){\begin{picture}(0.0,0.0)
                           \thicklines{\line(1,0){390.0}}
                         \end{picture}
                        }

      \put(365.0, 40.0){\begin{picture}(0.0,0.0)
                           \thinlines{\line(0,-1){5.0}}
                         \end{picture}
                        }

      \put(50.0, 40.0){\begin{picture}(0.0,0.0)
                           \thinlines{\line(0,-1){5.0}}
                         \end{picture}
                        }

      \put(20.0, 80.0){\begin{picture}(0.0,0.0)
                           \thicklines{\line(1,0){390.0}}
                         \end{picture}
                        }

   \end{picture}
 }\caption{$G''$ and $M$}\label{figure2}
}\end{figure}

Lemma~\ref{base} shows that with high probability, the left and last
parts of the motif region in a $\Theta(n,G,\alpha)$-sequence do not
change much.

\begin{lemma}\label{base}  With probability at most
$Q_0$,
%=(1-\alpha)^{2}-{2c^{v}\over 1-c}$,
a $\Theta(n,G,\alpha)$-sequence $S$ does not contain a stable motif
region.
\end{lemma}

\begin{proof}
The probability is $V_1=2(v-1)\alpha$ not to satisfy conditions (1)
and (2) of Definition~\ref{stable-def}. Consider condition (3).
Since every character of $\aleph(S)[1, m]$ (notice that $m=|G|$) has
probability at most $\alpha$ to mutate, by
Corollary~\ref{chernoff-lemma-a}, the probability is at most
$e^{-{1\over 3}\epsilon^2 r}$ that $\diff(G[1, h],G'[1, h])>
{\beta\over 2}=\alpha +\epsilon$. Let
$V_3=\sum_{r=v}^{\infty}e^{-{1\over 3}\epsilon^2 r}={c^v\over 1-c}$,
where $c=e^{-{1\over 3}\epsilon^2}$ as defined in
Definition~\ref{param-def}. Therefore, the probability is at most
$V_3$ that $\diff(G[1, h],G'[1, h])> {\beta\over 2}=\alpha
+\epsilon$ for some $h\in\{v, v+1,\cdots, m\}$.
% Section~\ref{parameters-subsect}.
Similarly we define $V_4=\sum_{r=v}^{\infty}e^{-{1\over 3}\epsilon^2
r}\le {c^v\over 1-c}$ for the probability on the right-hand side.
The probability is at most $V_4$ that $\diff(G[m-h, m],G'[m-h, m])>
{\beta\over 2}=\alpha +\epsilon$ for some $h\in\{v, v+1,\cdots,
m\}$. The probability that $S$ does not contain a stable motif
region is at most $V_1+V_3+V_4=Q_0$.
\end{proof}

\begin{definition} Assume that $Z_1=\{S_1',\cdots, S_{2k_1}\}$ contains
$S_{2i-1}'$ that contains a stable motif region. We fix such a
$S_{2i-1}'$.
\begin{itemize}
\item
Define $G_{L}=\aleph(S_{2i-1}')[1,d_0\log n-1]$ to be the left part
of the motif region $\aleph(S_{2i-1}')$.
%with $1\le i\le k_1$ that is stable.
\item
Define $G_{R}=\aleph(S_{2i-1}')[|G|-(d_0\log n)+1,|G|]$ to be the
right part of the motif region $\aleph(S_{2i-1}')$.
% with $1\le i\le k_1$ that is stable.
\end{itemize}
\end{definition}

Lemma~\ref{small-shift-lemma} shows that with high probability,
\phasetwo~ of algorithm \algmnam~extracts the correct motif regions
from the sequences in $Z_1$. It uses $G''$ to match $\aleph(S)$ in
another sequences $S$. The parameter $R$ gives a small probability
that the matched region between $G''$ and $S$ is not in $\aleph(S)$.

\begin{lemma}\label{small-shift-lemma}\scrod
\begin{enumerate}
\item
 Assume that $G_l$ and $G_r$ are  fixed sequences of length $d_0\log n$.
 Let $S$ be a
$\Theta(n,G,\alpha)$-sequence with $M\in \match(G_l, G_r, S)$ and
let $w_0$ be the number of characters of M that are not in the
region of $\aleph(S)$. Then the probability is  at most $R$ that
$w_0\ge 1$, where $R$ is defined in Definition~\ref{param-def}.
\item
 The probability is  at most
$Q_0$ that given a $\Theta(n,G,\alpha)$-sequence $S$,
$\match(G_{L},G_R, S)=\emptyset$.
\end{enumerate}
\end{lemma}

\begin{proof} Assume that $w_0\ge 1$. Let $w$ be the number of characters outside of $\aleph(S)$ on
the left of $M$, and let $w'$ be the number of characters outside of
$\aleph(S)$ on the right of $M$. Clearly, $w_0=w+w'$. Since $w_0\ge
1$, either $w\ge 1$ or $w'\ge 1$. See Figure~\ref{figure2}. Without
loss of generality, we assume $w\ge 1$.

Statement i: There are two cases.
%Without loss of generality, we assume that case (a) holds with $N^*$
%not equal to empty string.

Case (a): $1\le w<v$.  By Lemma~\ref{base-probability-lemma}, the
probability for this case is at most ${1\over t}$ for a fixed $w$.
The total probability for this case for $1\le w<v$ is at most
$\sum_{i=1}^{v-1}{1\over t^i}\le \sum_{i=1}^{\infty}{1\over
t^i}={1\over t-1}$.

Case (b): $v\le w$. By Lemma~\ref{base-probability-lemma}, the
probability is at most $e^{-{\epsilon^2\over 3}w}$ for a fixed $w$.
The total probability for $v\le w$ is at most $\sum_{w=v}^{\infty}
e^{-{\epsilon^2\over 3}w}={c^v\over 1-c}$.

The probability analysis is similar when $w'\ge 1$. Therefore, the
probability for this case is at most $R=({1\over t-1}+{c^v\over
1-c})$ for $w_0\ge 1$.

Statement ii: By Lemma~\ref{base}, with probability at most  $Q_0$,
 $S$ does not contain a stable motif region.
%satisfies the conditions of Lemma~\ref{base}.
%Since it has at most probability $R$ for $w_0\ge 1$.
%By the first case of this lemma,
Therefore, we have probability at most $Q_0$ that given a random
$\Theta(n,G,\alpha)$-sequence $S$, $\match(G_L, G_R,S)=\emptyset$.
\end{proof}

Lemma~\ref{large-algphabet-shift0-lemma} shows that we can use $G_l$
and $G_r$ to  extract most of the motif regions for the sequences in
$Z_2$ if
 $G'=G_{L}$ (recall that $G_{L}$ is defined right after Lemma~\ref{base}).

\begin{lemma}\label{large-algphabet-shift0-lemma} Assume that $G_l$ and $G_r$
are two sequences of length $d_0\log n$, and
 $G_i=\match(G_l, G_r, S_i'')$ for $S_i''\in Z_2=\{S_1'',\cdots, S_{k_2}''\}$ and $i=1,\cdots, k_2$  (recall that
 each sequence
$G_i$ is either an empty sequence or a sequence of the length
$|G_l|$).
\begin{enumerate}
\item
 If $G_l=G_{L}$, $G_r=G_R$, and there are no more than $yk_2$ ($y\in
 [0,1]$)
sequences $S_i''$ with $\roughleft_{S_i''}\not\in
[\LB(S_i'')-(v+u_2),\LB(S_i'')]$ or $\roughright_{S_i''}\not\in
[\RB(S_i''),\RB(S_i'')+(v+u_2)]$, then the probability is  at most
$e^{-{\epsilon^2 k_2\over 3}}$ that there are more than
$(Q_0+y+\epsilon)k_2$ sequences $G_i$ with $G_i=\emptyset$.
\item
For arbitrary $G_l$ and $G_r$, with probability at most
$e^{-{\epsilon^2 k_2\over 3}}$, $|\{i|G_i\not=\emptyset\ {\rm and}\
G_i\not=\aleph(S_i''), i=1,\cdots, k_2\}|> (R+\epsilon)k_2$, where
$R$ is defined in Definition~\ref{param-def}.
\end{enumerate}
\end{lemma}

\begin{proof} Recall that sequence $G_{1L}$ is selected right after
Lemma~\ref{base}.

Statement i: By Lemma~\ref{small-shift-lemma}, for every $S_i''\in
Z_2$, the probability is  at most $Q_0$ that $S_i''$ does not
contain a stable motif region $\aleph(S_i'')$. By
Corollary~\ref{chernoff-lemma-a}, we have probability at most
$e^{-{\epsilon^2 k_2\over 3}}$ that there are more than
$(Q_0+y+\epsilon)k_2$ sequences $G_i$ with $G_i=\emptyset$.

Statement ii:  By Lemma~\ref{small-shift-lemma}, the probability is
at most $R$ that $G_i\not=\aleph(S_i'')$.  By
Corollary~\ref{chernoff-lemma-a}, with probability  at most
$e^{-{\epsilon^2 k_2\over 3}}$, $|\{i|G_i\not=\aleph(S_i''),
i=1,\cdots, k_2 \}|> (R+\epsilon)k_2$.
\end{proof}

%\subsection{Proof of Main Theorem}

%We give the proof of Theorem~\ref{main-theorem}.

\begin{definition}\scrod
\begin{itemize}
\item
 Given two sequences $G_r$ and $G_r$, define

$M(G_r, G_r)=\{G_i'': G_i''=$Match$(G_l, G_r,
\roughleft_{S_i''},\roughright_{S_i''}, S_i'')$ $i=1,\cdots, k_2\}$.
\item
For a $\Theta_{\alpha}(n,G)$ sequence $S$, define $G_{S,L}$ to be
the $\aleph(S)[1,d_0\log n]$, which is the leftmost subsequence of
length $d_0\log n$ in the motif region of $S$.
\item
For a $\Theta_{\alpha}(n,G)$ sequence $S$, define  $G_{S,R}$ to be
the $\aleph(S)[m-d_0\log n+1,m]$, which is the rightmost subsequence
of length $d_0\log n$ in the motif region of $S$, where
$m=|G|=|\aleph(S)|$.
\end{itemize}
\end{definition}

the condition~\ref{V_0-condition} of Lemma~\ref{general-lemma}

\begin{lemma}\label{general-lemma}
Assume that we have the following conditions:
\begin{enumerate}
%\item
%With probability at most $p(n)$, Initial-Boundary($S_1,S_2)$ returns
%$(L_{S_1}, R_{S_1}, L_{S_2}, R_{S_2})$ with $L_{S_i}\not\in
%[\LB(S_i)-(v+u_1),\LB(S_i)]$
% or $R_{S_i}\not\in [\RB(S_i),\RB(S_i)+(v+u_1))]$ for $i=1,2$; and
%1. CollisionDetection$(S_1, U_1,S_2,U_2)$ returns
%$(L_{S_1}, R_{S_1}, L_{S_2}, R_{S_2})$ with $L_{S_i}\in
%[\LB(S_i)-L_0(S_i,L_{S_i},R_{S_i}),\LB(S_i)]$
% and $R_{S_i}\in [\RB(S_i),\RB(S_i)+L_0(S_i,L_{S_i},R_{S_i})]$ for $i=1,2$
% with probability at most $p_{selection}(n,L)$.
% has
%probability $p_{selection}(n,L)$ to provide collision for these
%$U_1$ and $U_2$ returned by Point-Selection($S_1,S_2,L$);
\item\label{varsigma1-condition}
For each $L$ with $0<L\le {|G|\over 2}$, with probability at most
$\varsigma_1(n)$, $L_{S_i}\not \in [\LB_{S_i}-2L, \LB_{S_i}+2L]$ and
$R_{S_i}\not \in [\RB_{S_i}-2L, \RB_{S_i}+2L]$ for $i=1,2$, where
$(L_{S_1}, R_{S_1}, L_{S_2}, R_{S_2})=$Collision-Detection($S_1,
U_1, S_2, U_2)$, $U_1=$Point-Selection$(S_1,L, [1,|S_1|])$, and
$U_2=$Point-Selection$(S_2,L, [1,|S_2|])$.

\item\label{varsigma2-condition}
For each $L$ with $0<L\le {|G|\over 2}$, if $S_1$ has
$\roughleft_{S_1}\not \in [\LB_{S_1}-L, \LB_{S_1}+L]$ and
$\roughright_{S_1}\not \in [\RB_{S_1}-L, \RB_{S_1}+L]$, then
 with probability at most
$\varsigma_2(n)$, $L_{S_i''}\not \in [\LB_{S_i''}-2L,
\LB_{S_i''}+2L]$ for $i=1,2$, where $(L_{S_1}, R_{S_1}, L_{S_i''},
R_{S_i''})=$Collision-Detection($S_1, U_1, S_i'', U_2)$,
$U_1=$Point-Selection$(S_1,L,
[\roughleft_{S_1}-2L,\roughleft_{S_1}+2L])\cup$
Point-Selection$(S_1,L,
[\roughright_{S_1}-2L,\roughright_{S_1}+2L])$, and
$U_2=$Point-Selection$(S_i'',L, [1,|S_i''|])$.

\item\label{P0-Q0-condition}
The inequality $(P_0+Q_0)<c_0$ holds for some constant $c_0<1$,
where $Q_0$ is defined at equation (\ref{Q0-def-eqn}) and
$P_0=\varsigma_1(n)+ {2(v+u_1)c^{v+u_1}\over (1-c)^2}+{c^v\over
1-c}+{1\over 5\cdot 2^x n}$.

\item\label{V_0-condition} The inequality
$1-2(Q_0+V_0+(R+2\epsilon))-(\alpha+\epsilon)>0$ holds, where
$V_0=(2(\varsigma_2(n)+ (v+u_1){c^{v+u_2}\over 1-c}+{c^{v}\over
1-c})+\epsilon)$.
\end{enumerate}

Then the algorithm generates a set of at most $k_2$ subsequences for
voting and votes a sequence $G'$ such that

(1) with probability at most $e^{-\Omega(k_1)}+e^{-\Omega(k_2)}$,
$|G'|\not=|G|$, and

(2) for each $1\le i\le |G|$, with probability at most
$e^{-\Omega(k_1)}+e^{-\Omega(k_2)}$, $G'[i]\not=G[i]$.
\end{lemma}

Before proving Lemma~\ref{small-shift-lemma}, we note that both
$\varsigma_1(n)$ and $\varsigma_2(n)$ is at most ${1\over 2^xn^3}$
for all of the three algorithms. They will be proved by
Lemma~\ref{collision-lemma} and Lemma~\ref{collision-lemma2} for the
case~\algtype=\sublinear, Lemma~\ref{collision-lemma-2A} and
Lemma~\ref{collision-all-select-lemma} for the
case~\algtype=\randomized, and Lemma~\ref{collision-lemma2-3A} for
the case~\algtype=\deterministic.

\begin{proof}
%Let $k=2k_1+k_2$ be the total number of input sequences.
%We also assume that $k\le d_1\log n$ as that $d_1\log n$ many input
%sequences will be enough to recover the true motif with high
%probability.

By Lemmas~\ref{initial-boundary-lemma}% and~\ref{collision-lemma}
, with probability at most $P_0=\varsigma_1(n)+
{2(v+u_1)c^{v+u_1}\over (1-c)^2}+{c^v\over 1-c}+{1\over 5\cdot 2^x
n}$, $\roughleft_{S_{2i-1}'}\not\in
[\LB(S_{2i-1}')-(v+u_1),\LB(S_{2i-1}')]$ or
$\roughright_{S_{2i-1}'}\not\in
[\RB(S_{2i-1}'),\RB(S_{2i-1}')+(v+u_1)]$.

 By
Lemma~\ref{select-median-lemma}, with probability at most
$P_a=e^{-{(0.5-P_0-\epsilon)^2k_1/3}}=e^{\Omega(k_1)}$, the
approximate motif length $l_{motif}$ is not in the range
$[|G|-2(v+u_1), |G|+2(v+u_1)]$.

 By Lemma~\ref{base}, with probability at
most $Q_0$, a $\Theta_{\alpha}(n, G)$ sequence does not contain a
stable motif region. Therefore, with probability at most
$P_1=(P_0+Q_0)^{k_1}$, the following statement is false.

(i) One of $S_{2i-1}'$ for $i=1,\cdots, k_1$ has
$\roughleft_{S_{2i-1}'}\in [\LB(S_{2i-1}')-(v+u_1),\LB(S_{2i-1}')]$,
$\roughright_{S_{2i-1}'}\in
[\RB(S_{2i-1}'),\RB(S_{2i-1}')+(v+u_1)]$, and has a stable motif
region.

By Lemma~\ref{left-boundary-detect-lemma}, with probability at most
$P_2=e^{-{\epsilon^2k_2\over 3}}$, there are more than
$(2(\varsigma_2(n)+ (v+u_1){c^{v+u_2}\over 1-c}+{c^{v}\over
1-c})+\epsilon)k_2$ sequences $S_i''$ with
$\roughleft_{S_i''}\not\in [\LB(S_i'')-(v+u_2),\LB(S_i'')]$ or
$\roughright_{S_i''}\not\in [\RB(S_i''),\LB(S_i'')+(v+u_2)]$. In
other words, with probability at most $P_2$, the following statement
is false:

(ii) There are no more than $V_0k_2$ sequences $S_i''$ with
$\roughleft_{S_i''}\not\in [\LB(S_i'')-(v+u_2),\LB(S_i'')]$ or
$\roughright_{S_i''}\not\in [\RB(S_i''),\RB(S_i'')+(v+u_2)]$, where
$V_0=(2(\varsigma_2(n)+ (v+u_1){c^{v+u_2}\over 1-c}+{c^{v}\over
1-c})+\epsilon)$.

%By Lemma~\ref{???}, with probability at most $k???$, there is a
%sequence that has rough left boundary with at least $v+u_1$ distance
%to its true left boundary.

Assume that Statement (ii) is true. By
Lemma~\ref{large-algphabet-shift0-lemma}, with probability at most
$P_3=c^{k_2}$, the following statement is false.

(iii) $M(G_L,G_R)$ contains at most $(Q_0+V_0+\epsilon)k_2$ empty
sequences.

We start from the rough left boundary $\roughleft_1$ of $S_1$ to
match the other left boundaries of $S_i''$ for $i=1,\cdots, k_2$.
There are totally at most $2(v+u_1)$ candidates to consider.

By Lemma~\ref{large-algphabet-shift0-lemma},  if $M(G_l, G_r)$,
which consists $k_2$ matched regions,  has at most
$(Q_0+V_0+\epsilon)k_2$ empty sequences, then it has more than
$(R+\epsilon)k_2$ from non-motif regions with probability at most
$P_4=2(v+u_1)e^{-{\epsilon^2 k_2\over 3}}$.
% to fail this property as we have at most $2(v+u_1)$ candidates.
After the pattern is fixed, those events in the matching are
considered to be independent each other. This is why we can apply
the Chernoff bound to deal with them. So, the probability is at most
$P_4$, the following statement is false.

(iv). If  $M(G_l, G_r)$ contains at most $(Q_0+V_0+\epsilon)k_2$
empty sequences, then $M(G_l, G_r)$ contains at most
$(Q_0+V_0+\epsilon+(R+\epsilon))k_2=(Q_0+V_0+(R+2\epsilon))k_2$
elements not from motif regions $\{\aleph(S_i''): 1\le i\le k_2\}$.

Therefore, with probability at most $P_1+P_2+P_3+P_4=
e^{-\Omega(k_1)}+e^{-\Omega(k_2)}$, the sequences are not ready for
voting in the next phase, which means the following two conditions
are satisfied:

(a). There exists $G_l$ and $G_r$ generated by the algorithm such
that %$|M(G_l,G_r)|$ contains at most $(Q_0+V_0+\epsilon)k_2$ empty sequences, and
$M(G_l,G_r)$ contains at most $(Q_0+V_0+(R+2\epsilon))k_2$ elements
not from motif regions $\{\aleph(S_i''): 1\le i\le k_2\}$.

(b). For every $G_l$ and $G_r$ that  $M(G_l, G_r)$ contains at most
$(Q_0+V_0+\epsilon)k_2$ empty sequences generated by the algorithm,
$M(G_l, G_r)$ contains at most
$(Q_0+V_0+\epsilon+(R+\epsilon))k_2=(Q_0+V_0+(R+2\epsilon))k_2$
elements not from motif regions $\{\aleph(S_i''): 1\le i\le k_2\}$.

Statement (1): For a $M(G_l, G_r)$ with at most
$(Q_0+V_0+(R+2\epsilon))k_2$ elements not from motif regions
$\{\aleph(S_i''): 1\le i\le k_2\}$, we still have
$k_2-(Q_0+V_0+(R+2\epsilon))k_2$ elements in $M(G_l, G_r)$ from
motif regions $\{\aleph(S_i''): 1\le i\le k_2\}$.
%By the choice of $v$ to satisfy inequality~(\ref{v-set-ineqn1}),
By by the condition~(\ref{V_0-condition}) in this lemma, we have
$k_2-(Q_0+V_0+(R+2\epsilon))k_2>(Q_0+V_0+(R+2\epsilon))k_2$.
Therefore, $|G'|$ is selected to be the length of $G$ in the
\phasethree().

Statement (2): For a $M(G_l, G_r)=\{G_1'',\cdots, G_{k_2}''\}$ with
at most $(Q_0+V_0+(R+2\epsilon))k_2$ elements not from motif regions
$\{\aleph(S_i''): 1\le i\le k_2\}$, we still have
$k_2-(Q_0+V_0+(R+2\epsilon))k_2$ elements in $M(G_l, G_r)$ from
motif regions $\{\aleph(S_i''): 1\le i\le k_2\}$.
%For the voting part to
%recover each character in motif, each character has probability at
%most $\alpha$ to mutate.
By Corollary~\ref{chernoff-lemma-a}, with probability at most
$e^{-{\epsilon^2 k_2\over 3}}$ there are more than
$(\alpha+\epsilon)k_2$ characters are mutated in the same position
among all $k_2$ the motif regions for the sequences in $Z_2$. We
have that
$k_2-(Q_0+V_0+(R+2\epsilon))k_2-(\alpha+\epsilon)k_2>(Q_0+V_0+(R+2\epsilon))k_2$
by the condition~(\ref{V_0-condition}) in this lemma.
% by inequality~(\ref{v-set-ineqn1}).
We let $G'[j]$ be the most
frequent character among $G_1''[j],\cdots, G_{k_2}''[j]$ in
\phasethree.
 Therefore, with probability at most $e^{-{\Omega(k_1)}}+e^{-{\Omega(k_2)}}$, $G'[j]\not=G[j]$.

%The computational time of the entire algorithm follows from Lemma~\ref{total-time-lemma}.
%By the selection of range of $L$, the computational time is
%determined by the function $({n\over L}M+L^2)\log n$. The second
%phase takes $O((v+u_1)(v+u_2)\log n)=O((\log\log n)(\log n))$ time
%to match two motif regions between two sequences
%Since $n^{2\over 5}$ is the maximal for $L$, the complexity is
%$O({n\over \sqrt{h}}(\log n)^{3\over 2}+h^2)\log n)$ time, where $n$
%is the longest length of any input sequences, and
%$h=\min(|G|,n^{2\over 5})$. This selection of range for $L$ can make
%the total time complexity sublinear.
\end{proof}

We will use multiple variable functions to characterize the
computational time for three algorithms. In order to unify the
complexity analysis of three algorithm, we introduce the following
notation.

\begin{definition}
A function $T(x,y): N\times N\rightarrow N$ is monotonic if it is
monotonic on both variables. If for arbitrary positive constants
$c_1$ and $c_2$,  $T(c_1x, c_2y)\le cT(x,y)$ for some positive
constant $c$, then $T(x,y)$ is {\it slow}.
\end{definition}

\begin{lemma}\label{total-time-lemma} Assume that $T(x,y)$, $s(n,L)$
and $g(n,l)$ are monotonic slow functions.
 Assume that Collision-Detection($S_1, U_1, S_2, U_2$) returns
 the result in time $t(n, ||U_1||+||U_2||)$ time and the Point-Selection($S_1, S_2,L)$) selects $s(n,L)$ positions in
 $g(n,L)$ time.
Assume that with probability at most $\varphi(n)$, the function does
not stop Initial-Boundaries() does not stop when $L\le |G|/4$, and
$||U_{S_{2i-1}'}||+||U_{S_j''}||$ in the algorithm \algmnam~ is no
more than $f(n,|G|)$.

 Then with probability at most $k_1\varphi(n)$,
the entire algorithm \algmnam~  does not stop in the time complexity
$(O(k_1(\sum_{i=1}^{i_0}(T(n,s(n,{n\over 2^in^{2/5}}))+g(n,{n\over
2^{i_0}n^{2/5}})))+k_1h^2\log n+k_1k_2 t(n,f(n, |G|))+h^2\log
n)+k_1k_2(\log n)(\log\log n)), O(k_2))$, where $i_0$ is the largest
$j$ such that ${n\over 2^jn^{2/5}}\le \min(n^{2/5},|G|)$ and
$h=\min(n^{2/5}, |G|)$.
\end{lemma}

\begin{proof}
The function Initial-Boundaries()is executed $k_1$ times. According
to the condition that with probability at most $\varphi(n)$, the
function does not stop Initial-Boundaries(.) does not stop when
$L\le |G|/4$, we have the fact that with probability at most
$k_1\varphi(n)$, one of those executions of Initial-Boundaries(.)
does not stop when $L\le |G|/4$.
%By Lemma~\ref{collision-lemma},

In the rest of the proof, we assume that all executions of
Initial-Boundaries(.) stops when $L\le |G|/4$.

When $L=O(h)$, we detect rough left and right motif boundaries and
run Improve-Boundaries(), which takes $O(h^2\log n)$ time. It takes
$O(\sum_{i=1}^{i_0}(T(n,s(n,{n\over 2^in^{2/5}}))+g(n,{n\over
2^in^{2/5}}) +h^2\log n)$ time to run Initial-Boundaries($S_{2i-1}',
S_{2i}'$) one time for one pair ($S_{2i-1}', S_{2i}'$) in $Z_1$. It
takes $O(k_1(\sum_{i=1}^{i_0}(t(n,s(n,{n\over
2^in^{2/5}}))+g(n,{n\over 2^in^{2/5}})+k_1h^2\log n )$ time to run
Initial-Boundaries($S_{2i-1}', S_{2i}'$) one time for all pairs
($S_{2i-1}', S_{2i}'$) in $Z_1$.

It takes $k_2(t(n,f(n, |G|))+h^2\log n)$ time to find the rough
boundaries for all sequences in $Z_2$ with a fixed sequence $S$ from
$Z_1$ by executing the for loop ``For each $S_j''\in Z_2$" in the
algorithm \algmnam. It takes $k_1k_2(t(n,f(n, |G|))+h^2\log n)$ time
to find the rough boundaries for all sequences in $Z_2$ via all
sequences $S_{2i-1}'$ from $Z_1$ through for loop ``For each
$S_j''\in Z_2$" in the algorithm \algmnam.

Recall that parameters $v$ and $u_1$ are constants, and $u_2$ is
$O(\log\log n)$. Calling Match($G_l,G_r,S_i''$) takes $O((v+u_2)\log
n)$ time for each $S_i''\in Z_2$. The total times for calling
Match($G_l,G_r,S_i''$) is $O(k_1k_2(v+u_1)(v+u_2)\log
n)=O(k_1k_2(\log n)(\log\log n))$.

The voting part takes $O(k_2)$ time for executing voting for
recovering one character in motif.
% which is sufficient to output result with high probability.
%By the selection of range of $L$, the computational time is
%determined by the function $({n\over L}M+L^2)\log n$. The second
%phase takes $O((v+u_1)(v+u_2)\log n)=O((\log\log n)(\log n))$ time
%to match two motif regions between two sequences
%Since $n^{2\over 5}$ is the maximal for $L$, the complexity is
%$O({n\over \sqrt{h}}(\log n)^{3\over 2}+h^2)\log n)$ time, where $n$
%is the longest length of any input sequences, and
%$h=\min(|G|,n^{2\over 5})$. This selection of range for $L$ can make
%the total time complexity sublinear.
\end{proof}

\subsection{Randomized Algorithms for Motif Detection}

In this section, we present two randomized algorithms for motif
detection. The first one is a sublinear time algorithm that can
handle ${1\over (\log n)^{2+\mu}}$ mutation, and the second one is a
super-linear time algorithm that can handle $\Omega(1)$ mutation.
They also share some common functions.

\begin{lemma}\label{alpha-lemma} Let $c$ be a constant in $(0,1)$.
Assume $m$ and $n$ are two non-negative integer with $m\le n$. Then
for every integer $m_1$ with $0\le m_1\le {\delta_c m\over \ln n}$,
${n\choose m_1} c^m \le e^{(m\ln c)/2}$, where constant
$\delta_c={-\ln c\over 2}$ as defined in Definition~\ref{param-def}.
\end{lemma}

\begin{proof} We have the inequalities
\begin{eqnarray}
 {n\choose m_1} c^m
 %&\le& {e^{m_1}n^{m_1}\over
%m_1^{m_1}}c^m \\
%&=&({en\over m_1})^{m_1}c^m\\
&\le& n^{m_1}c^m\\
&=& e^{m_1\ln n}c^m\\
&\le& e^{{\delta_c m\over \ln n}\ln n}c^m\\
&=& e^{\delta_cm}e^{m\ln c}\label{final-eqn1}\\
 &=&e^{(m\ln c)/2\label{final2-eqn1}}
\end{eqnarray}
%For the transition from (\ref{final-eqn1}) to (\ref{final2-eqn1}),
%it is because of $0\le m_1\le {\delta_c m\over \log n}$ and
%inequality $(\ref{alpha-eqn})$ (see Definition~\ref{param-def} for
%$\delta_c$).
\end{proof}

\begin{lemma}\label{spread-lemma}
Let  $S=U\cup V$ be a set of $n$ elements with $U\cap V=\emptyset$.
Assume that $x_1,\cdots, x_m$ are $m$ random elements in $S$. Then
with probability at most ${||U||\choose m_1} ({||V||+m_1\over
n})^m$, the list $x_1,\cdots, x_m$ contains at most $m_1$ different
elements from $U$ (in other words, $||\{\ x_1,\cdots, x_m\}\cap
U||\le m_1$).
\end{lemma}

\begin{proof}
For a subset $S'\subseteq S$ with $|S'|=m_1$, the probability is at
most $({m_1\over n})^m$ that all elements $x_1,\cdots, x_m$ are in
$S'$. For every subset $X\subseteq S$ with $|X|\le m_1$, there
exists another subset $S'\subseteq S$ such that $|S'|=m_1$. We have
that $\prob[||\{\ x_1,\cdots, x_m\}\cap U||\le m_1]\le \prob[\{
x_1,\cdots, x_m\}\cap U\subseteq U' {\rm \ for \ some}\ U'\subseteq
U \ {\rm with \ } ||U'||=m_1]$. There are ${||U||\choose m_1}$
subsets of $U$ with size $m_1$. We have the probability at most
${||U||\choose m_1} ({||V||+m_1\over n})^m$ that $x_1,\cdots, x_m$
contains at most $m_1$ different elements in $U$.
\end{proof}

\begin{lemma}\label{intersection-lemma}
Let $\delta$ be the same as that in Lemma~\ref{alpha-lemma}.
%Let $\delta_1$ be the constant???
Let $\beta$ be a constant in $(0,1)$ and $c=1-{\beta\over 2}$.
%Let $m_1\log n\le \delta_1m$
Let $m_1\le {\delta_c m\over \ln \beta n}$  and $m\le
n^{1-\epsilon}$ for some fixed $\epsilon>0$. Let $S_1$ and $S_2$ be
two sets of $n$ elements with $|S_1\cap S_2|\ge \beta n$ and $C$ be
a set of size $|C|\le \gamma m_1$ for some constant $\gamma\in
(0,1)$. Then for all large $n$, with probability is at most
$2e^{-{(1-\gamma) m_1 m\over n}}$, we have $(A-C)\cap
(B-C)=\emptyset$, where $A=\{x_1,\cdots, x_m\}$ and $B=\{y_1,\cdots,
y_m\}$ are two sets, which may have multiplicities, of $m$ random
elements from $S_1$ and $S_2$, respectively.
\end{lemma}

\begin{proof} In the entire proof of this lemma, we always assume
that $n$ is sufficiently large.
%Since  $m_1\le \delta_1m$ and $m\le n^{1-\epsilon}$ for some fixed
%$\epsilon>0$, we assume that $m_1\le {\beta n\over 2}$. We assume
%that $A$ has at least $m_1=\ceiling{\delta_1 m}$ different elements
%from $S_1\cap S_2$.
We are going to give an upper bound about the probability that $B$
does not contain any element in $A-C$. For each element $y_i\in B$,
with probability at most $1-{m_1\over n}$ that $y_i$ is not in $A$.
Therefore, the probability is at most $(1-{||A||-||C||\over n})^m$
that $B$ does not contain any element in $A-C$.

By Lemma~\ref{spread-lemma}, the probability is at most ${\beta
n\choose m_1} ({(1-\beta)n+m_1\over n})^m$ that $||A\cap (S_1\cap
S_2)||\le m_1$. We have the inequalities

\begin{eqnarray}
&&\prob[(A-C)\cap (B-C)=\emptyset]\\
&=&\prob[(A-C)\cap (B-C)=\emptyset|\ \ ||A\cap (S_1\cap S_2)||\ge
m_1]\cdot \prob[||A\cap (S_1\cap
S_2)||\ge m_1]+ \\
&&\prob[(A-C)\cap (B-C)=\emptyset|\ \ ||A\cap (S_1\cap S_2)||<
m_1]\cdot \prob[|A\cap (S_1\cap
S_2)|< m_1]\label{ineq1}\\
&\le& \prob[(A-C)\cap (B-C)=\emptyset|\ \ ||A\cap (S_1\cap S_2)||\ge
m_1]+ \prob[||A\cap (S_1\cap
S_2)||< m_1]\label{ineq2}\\
&\le& (1-{||(A\cap S_1\cap S_2)||-||C||\over n})^m+{\beta
n\choose m_1} ({(1-\beta)n+m_1\over n})^m\label{ineq3}\\
&\le& (1-{(1-\gamma)m_1\over n})^m+{\beta
n\choose m_1} ({(1-\beta)n+m_1\over n})^m\label{ineq3b}\\
&\le& e^{-{(1-\gamma) m_1m\over n}}+{\beta n\choose m_1}
({(1-\beta)n+m_1\over n})^m\label{ineq4}\\
&\le& e^{-{(1-\gamma) m_1m\over n}}+{\beta n\choose m_1}
(1-{\beta\over 2})^m\label{ineq5}\\
&\le& e^{-{(1-\gamma) m_1m\over n}}+e^{(m\ln c)/2}\label{ineq6}\\
&\le& 2e^{-{(1-\gamma) m_1 m\over n}}.\label{ineq7}
\end{eqnarray}
The inequality $(1-{(1-\gamma)m_1\over n})^m\le e^{-{(1-\gamma)
m_1m\over n}}$, which is used from (\ref{ineq3b}) to (\ref{ineq4}),
follows from the fact that $1-x\le e^{-x}$. The transition from
(\ref{ineq4}) to (\ref{ineq5}) follows from the fact ${m_1\over
n}\le {\beta\over 2}$ since $m_1=o(n)$ according to the conditions
of the lemma.

It is easy to see that ${2(1-\gamma)m_1m\over -m\ln
c}={2(1-\gamma)m_1\over -\ln c}\le n$ for all large $n$. Thus,
${(1-\gamma)m_1m\over n}\ge (m\ln c)/2$ (note that $\ln c<0$ as
$c\in (0,1)$). Thus, by Lemma~\ref{alpha-lemma},  ${\beta n\choose
m_1} (1-{\beta\over 2})^m\le e^{m\ln c/2} \le e^{-{(1-\gamma)
m_1m\over n}}$. This is why we have the transition from
(\ref{ineq6}) to (\ref{ineq7}). Therefore, $\prob[(A-C)\cap
(B-C)=\emptyset]\le 2e^{-{(1-\gamma) m_1m\over n}}$.
\end{proof}

\subsubsection{Sublinear Time Algorithm for ${1\over (\log
n)^{2+\mu}}$ Mutation Rate}

In this section, we give an algorithm for the case with at most
${1\over (\log n)^{2+\mu}}$ mutation rate. The performance of the
algorithm is stated in Theorem~\ref{main-theorem1}.

%\begin{corollary}Assume that $\tau $ and $\mu$ are fixed numbers in $(0,1)$
%and the alphabet size $t$ is at least $4$. Then there exist
%constants $c_0$ and $c_1$ such that
% if the length of the motif $G$ is at least $c_0\log n$ and the number of input $\Theta(n,G,o({1\over (\log
%n)^{2+\mu}}))$-sequences is at least $c_1\log n$, then algorithm
%\algmnam~has time complexity $(O({n\over \sqrt{h}}(\log n)^{3\over
%2}+h^2, O(k))$, and outputs $G$ with probability at least ${3\over
%4}$, where $n$ is the longest length of any input sequences, and
%$h=\min(|G|,n^{2\over 5})$.
%\end{corollary}

%\begin{proof}
%It follows from Theorem~\ref{main-theorem1}.
%\end{proof}

%\subsubsection{Analysis of the Sublinear Time Algorithm}

\begin{definition}
A position $p$ in the motif region $\aleph(S)$ of an input sequence
$S$ is {\it damaged} if there exists at least one  mutation in
$S[p,p+d_0\log n-1]$.
\end{definition}

\begin{lemma}\label{affected-points-lemma} Assume that $\alpha
L=(\log n)^{1+\Omega(1)}$. With probability at most $e^{-(\log
n)^{1+\Omega(1)}}$,  there are more than ${M_1\over (\log
n)^{\Omega(1)}}$ positions  that  are from the $M$ sampled positions
in an interval of length $L$ and are damaged.
\end{lemma}

\begin{proof} By Theorem~\ref{ourchernoff-theorem},
with probability at most $P_1=2^{-\alpha L}$ (let $\delta=2$), there
are more than $3\alpha L$ mutation in an interval of length $L$.
Therefore, with probability at most $2^{-\alpha L}=e^{-(\log
n)^{1+\Omega(1)}}$, there are more than $3\alpha L\log n$ positions
are damaged. Therefore, each random position in an interval of
length $L$ has at most probability ${3\alpha L\log n\over L}=3\alpha
\log n$ to be damaged.

% Let $\gamma={1\over (\log n)^{1+\delta_1}}$
%for a very small $\delta_1$.
Since $\alpha=({1\over (\log n)^{2+\Omega(1)}})$ and $M$ positions
are sampled,  by Theorem~\ref{ourchernoff-theorem}, with probability
at most $P_2={2^{-(3\alpha \log n) M}}=e^{-(\log n)^{1+\Omega(1)}}$
(let $\delta=2$), the number of damaged positions sampled in an
interval of length $L$ is more than $(1+\delta)3\alpha \log
n)M=(9\alpha \log n)M={M_1\over (\log n)^{\Omega(1)}}$. Thus, with
total probability at most $P_1+P_2=e^{-(\log n)^{1+\Omega(1)}}$,
there are more than ${M_1\over (\log n)^{\Omega(1)}}$ damaged
positions that are from the $M$ sampled positions in an interval of
length $L$.
\end{proof}

\begin{definition}
Let $A$ be a set of positions in an input sequence $S$ with
$\aleph(S)=[i,j]$. Let $A(S,\aleph(S))=A\cap [i,j]$.
\end{definition}

\begin{lemma}\label{collision-lemma} Assume that  $|G|\ge \thresholdL$ and $d_0\log n\le L\le |G|/2$.
Let $I_1$ be a union of intervals that include $[\LB(S_1)-2L,
\LB(S_1)+2L]$ and $[\RB(S_1)-2L, \RB(S_1)+2L]$. Let
$U_1=$Point-Selection$(S_1,L, I_1)$,
$U_2=$Point-Selection$(S_2,L,[1,|S_2|])$,   and
  $(L_{S_1},R_{S_1},L_{S_2},R_{S_2})=$Collision-Detection$(S_1,U_1,S_2,,U_2)$.
Then
\begin{enumerate}
\item
With probability at most ${1\over 2^{x}n^{3}}$, the left rough
boundary $\L_{S_1}$ has at most $2L$ distance from $\LB(S_1)$ and
the left rough boundary $L_{S_2}$ has at most $2L$ distance from
$\LB(S_2)$.

\item
 With probability at most ${1\over 2^{x}n^{3}}$, the right rough boundary $R_{S_1}$ has at most
$2L$ distance from $\RB(S_1)$; and the right boundary of $R_{S_2}$
has at most $2L$ distance from $\RB(S_2)$.
\end{enumerate}
\end{lemma}

\begin{proof}
We prove the following two statements which imply the lemma.
\begin{enumerate}
\item
With probability at most ${1\over 2^{x}n^{3}}$, there is no
intervals $A_i$ from $S_1$ and $B_j$ from $S_2$ such that (1)
$|A_i(S_1,\aleph(S_1))\cap B_j(S_2,\aleph(S_2))|$ is at least
${L\over 2}$; (2) the left boundary of $S_1$ has at most $2L$
distance from $A_i$; (3) the left boundary of $S_2$ has at most $2L$
distance from $B_j$; and (4) there is collision between the sampled
positions in $A_i$ and $B_j$.

\item
 With probability at most ${1\over
2^{x}n^{3}}$, there is no intervals $A_i$ from $S_1$ and $B_j$ from
$S_2$ such that (1) $|A_i(S_1,\aleph(S_1))\cap
B_j(S_2,\aleph(S_2))|$ is at least ${L\over 2}$; (2) the right
boundary of $S_1$ has at most $2L$ distance from $A_i$; (3) the
right boundary of $S_2$ has at most $2L$ distance from $B_j$; and
(4) there is collision between the sampled positions in $A_i$ and
$B_j$.
\end{enumerate}
We only prove the statement i. The proof for statement ii is similar
to that for statement i.
%We assume that the length of motif $G$ is at least $3L$.
Note that $L$ goes down by half each cycle in the algorithm. Assume
that $L$ satisfies the condition of this lemma.

Select $A_i$ from $S_1$ and $B_j$ from $S_2$ to be the first pair of
intervals with $||A_i(S_1,\aleph(S_1))\cap B_j(S_2,\aleph(S_2))||\ge
{L\over 2}$. It is easy to see that such a pair exists and both have
distance from the left boundary with distance at most $2L$. This is
because when an leftmost interval of length $L$ is fully inside the
motif region of the first sequence, we can always find the second
interval from the second sequence with intersection of length at
least ${L\over 2}$.

 Replace $m$ by $M(L)$, $m_1$ by $M_1(L)$ (see Definition~\ref{M-M1-def}), and $n$ by $L$ to apply
Lemma~\ref{intersection-lemma}. We also let $C$ be the set of
damaged positions  affected by the mutated positions. With
probability at most $o({1\over 2^xn^3})$, $C$ has size more than
$\Omega(M_1(L))$ by Lemma~\ref{affected-points-lemma}. With
probability at most $o({1\over 2^xn^3})$, there is an no
intersection $A_i$ from $S_1$ and $B_j$ from $S_2$.
\end{proof}

\begin{lemma}\label{collision-lemma2} Assume that   $|G|<\thresholdL$ and $L$ is an integer with $d_0\log n\le L\le |G|/2$.
Let $I_1$ be a union of intervals that include $[\LB(S_1)-2L,
\LB(S_1)+2L]$ and $[\RB(S_1)-2L, \RB(S_1)+2L]$. Let
$U_1=$Point-Selection$(S_1,L, I_1)$,
$U_2=$Point-Selection$(S_2,L,[1,|S_2|])$,   and
  $(L_{S_1},R_{S_1},L_{S_2},R_{S_2})=$Collision-Detection$(S_1,U_1,S_2,,U_2)$.
 Then
\begin{enumerate}
\item
With probability at most ${1\over 2^{x}n^{3}}$, the left rough
boundary $\L_{S_1}$ has at most $|G|/4$ distance from $\LB(S_1)$ and
the left rough boundary $L_{S_2}$ has at most $|G|/4$ distance from
$\LB(S_2)$.

\item
 With probability at most ${1\over 2^{x}n^{3}}$, the right rough boundary $R_{S_1}$ has at most
$|G|/4$ distance from $\RB(S_1)$; and the right boundary of
$R_{S_2}$ has at most $|G|/4$ distance from $\RB(S_2)$.
\end{enumerate}
\end{lemma}

\begin{proof}
For two sequences $S_1$ and $S_2$, it is easy to see that there a
common position in both motif regions of the two sequences such that
there is no mutation in the next $d_0\log n$ characters with high
probability. This is because that mutation probability is small.

By Theorem~\ref{ourchernoff-theorem}, with probability at most
$P_{l,1}=2^{-\alpha |G|/4}$ (let $\delta=2$), there are more than
$3\alpha
 {|G|\over 4}$ mutated characters in the interval $\aleph(S_i)[1, {|G|\over 4}]$ for
$i=1,2$. Therefore, with probability at most $2^{-\alpha
|G|/4}=e^{-(\log n)^{1+\Omega(1)}}$, there are more than $3\alpha
{|G|\over 4}\log n$ positions are damaged in  $\aleph(S_i)[1,
{|G|\over 4}]$.

% Let $\gamma={1\over (\log n)^{1+\delta_1}}$
%for a very small $\delta_1$.
Since the mutation probability is $\alpha=({1\over (\log
n)^{2+\Omega(1)}})$ and $M(L)$ positions are sampled, with
probability at most $P_{l,2}={2^{-(3\alpha d_0\log n)  {|G|\over
4}}}=e^{-(\log n)^{1+\Omega(1)}}$ (with $\delta=2$), the number of
damaged positions is more than $((5\alpha d_0\log n) {|G|\over
4})={|G|\over (\log n)^{\Omega(1)}}$ by
Theorem~\ref{ourchernoff-theorem}. The probability is
$P_l=P_{l,1}+P_{l,2}=e^{-(\log n)^{1+\Omega(1)}}$ that left side has
more than $((5\alpha d_0\log n) {|G|\over 4})={|G|\over (\log
n)^{\Omega(1)}}$ damaged positions.

We have similar $P_r=P_{r,1}+P_{r,2}=e^{-(\log n)^{1+\Omega(1)}}$
probability for the right side for more than $((5\alpha d_0\log n)
{|G|\over 4})={|G|\over (\log n)^{\Omega(1)}}$ damaged positions in
$\aleph(S_i)[{3|G|\over 4}-1, |G|]$.

Now we assume that left side has more than $((5\alpha d_0\log n)
{|G|\over 4})={|G|\over (\log n)^{\Omega(1)}}$ damaged positions and
the right side for more than $((5\alpha d_0\log n) {|G|\over
4})={|G|\over (\log n)^{\Omega(1)}}$ damaged positions in
$\aleph(S_i)[{3|G|\over 4}-1, |G|]$. Since each position in each
interval of length $L$ is selected in
Point-Selection$(S_1,S_2,L)$???, it is easy to verify the
conclusions of this lemma.
\end{proof}

\begin{lemma}\label{collision-time-lemma1}
For the case~\algtype=\sublinear, we have
\begin{enumerate}
\item
 CollisionDetection($S_1, U_1, S_2, U_2$)  takes
$t(n, ||U_1||+||U_2||)=O((||U_1||+||U_2||)\log n)$ time.
\item
Point-Selection($S_1,L, [1,|S_1|]$) selects $s(n,L)=O(({n\over
L})M(L))$ positions in $g(n,L)=O(s(n,L))$ time if $L\ge\thresholdL$.
\item
Point-Selection($S_1,L, [1,|S_1|]$) selects $s(n,L)=O(n)$ positions
in $g(n,L)=O(n)$ time if $L<\thresholdL$.
%\item
%Initial-Rough-Boundary$(S_1, S_2, L)$ takes $O(\log n)$ time.
\item
$||U_{S_{2i-1}'}||+||U_{S_j''}||$ in the algorithm \algmnam~ is no
more than $f(n,|G|)=O(M(|G|)+{n\over |G|}M(|G|))$.
\item
With probability at most ${k\over 2^xn^3}$, the  algorithm \algmnam~
does not stop in $(O(k({n\over \sqrt{h}}(\log n)^{5\over 2}+h^2\log
n)), O(k))$ time.
\end{enumerate}
\end{lemma}

\begin{proof}
Statement i. The parameter $\omega_{\sublinear}$ is set to be $0$ in
the Collision-Detection. It follows from the time complexity of
bucket sorting, which is described in standard algorithm textbooks.

Statements ii and iii. They follows from the implementation of
Point-Selection().

Statement iv. It follows from the choice of Point-Selection(.) for
the sublinear time algorithm at Recover-Motif(.).
% Only the binary search is needed and takes $O(\log n)$ time.

Statement v. It follows from Lemma~\ref{collision-lemma2},
Lemma~\ref{collision-lemma}, Lemma~\ref{total-time-lemma} and
Statements i, ii, and iii, and iv.
\end{proof}

We give the proof for Theorem~\ref{main-theorem1}.

\begin{proof}[Theorem~\ref{main-theorem1}]
The computational time part of this theorem follows from
Lemma~\ref{collision-time-lemma1}.
% and Lemma~\ref{total-time-lemma}.

By Lemma~\ref{collision-lemma}, Lemma~\ref{collision-lemma2}, we can
let $\varsigma_1(n)={1\over 2^xn^3}\le \varsigma_0$ for the
probability bound $\varsigma_1(n)$ in the
condition~(\ref{varsigma1-condition}) of Lemma~\ref{general-lemma}.

By Lemma~\ref{collision-lemma}, Lemma~\ref{collision-lemma2}, we can
let $\varsigma_2(n)={1\over 2^xn^3}\le \varsigma_0$ for the
probability bound $\varsigma_1(n)$ in the
condition~(\ref{varsigma2-condition}) of Lemma~\ref{general-lemma}.

By inequality~(\ref{support-P0-Q0-inequality}),
%(\ref{v-alpha0-ineqn}),
the condition~(\ref{P0-Q0-condition}) of Lemma~\ref{general-lemma}
is satisfied.

 By inequality~(\ref{v-set-ineqn1}), we know that the
condition~(\ref{V_0-condition}) of Lemma~\ref{general-lemma} can be
satisfied.

 The failure
probability part of this theorem follows from
%Lemma~\ref{collision-lemma}, Lemma~\ref{collision-lemma2},
Lemma~\ref{v+u-lemma2}, and Lemma~\ref{general-lemma} by using the
fact that $k_1,k_2$, and $k$ are of the same order (see
equation~(\ref{Z1-Z2-eqn})).
\end{proof}

\subsubsection{Randomized Algorithm for $\Omega(1)$ Mutation Rate} In
this section, we give an algorithm for the case with $\Omega(1)$
mutation rate.  The performance of the algorithm is stated in
Theorem~\ref{main-theorem2}.

%The parameter $\omega$ is set to be $\beta$ in the
%Collision-Detection.

%\subsubsection{Analysis of the Second Algorithm}

\begin{lemma}\label{collision-lemma-2A} Assume that  $d_0\log n\le L\le  |G|/2$
  and $|G|\ge \thresholdL$. Let $I_1$ be a union of intervals that include $[\LB(S_1)-2L,
\LB(S_1)+2L]$ and $[\RB(S_1)-2L, \RB(S_1)+2L]$. Let
$U_1=$Point-Selection$(S_1,L, I_1)$,
$U_2=$Point-Selection$(S_2,L,[1,|S_2|])$,   and
  $(L_{S_1},R_{S_1},L_{S_2},R_{S_2})=$Collision-Detection$(S_1,U_1,S_2,,U_2)$.
  Then
\begin{enumerate}
\item
With probability at most ${1\over 2^{x}n^{3}}$, the left rough
boundary $\L_{S_1}$ has at most $2L$ distance from $\LB(S_1)$ and
the left rough boundary $L_{S_2}$ has at most $2L$ distance from
$\LB(S_2)$.

\item
 With probability at most ${1\over 2^{x}n^{3}}$, the right rough boundary $R_{S_1}$ has at most
$2L$ distance from $\RB(S_1)$; and the right boundary of $R_{S_2}$
has at most $2L$ distance from $\RB(S_2)$.
\end{enumerate}
\end{lemma}

\begin{proof}
We prove the following two statements which imply the lemma.

\begin{enumerate}
\item
With probability at most ${1\over 2^{x}n^{3}}$, there is no
intervals $A_i$ from $S_1$ and $B_j$ from $S_2$ such that (1)
$||A_i(S_1,\aleph(S_1))\cap B_j(S_2,\aleph(S_2))||$ is at least
${L\over 2}$; (2) The left boundary of $S_1$ has at most $2L$
distance from $A_i$; (3) The left boundary of $S_2$ has at most $2L$
distance from $B_j$; and (4) There is collision between the sampled
positions in $A_i$ and $B_j$.

\item
 With probability at most ${1\over
2^{x}n^{3}}$, there is no intervals $A_i$ from $S_1$ and $B_j$ from
$S_2$ such that (1) $||A_i(S_1,\aleph(S_1))\cap
B_j(S_2,\aleph(S_2))||$ is at least ${L\over 2}$; (2) The right
boundary of $S_1$ has at most $2L$ distance from $A_i$; (3) The
right boundary of $S_2$ has at most $2L$ distance from $B_j$; and
(4) There is collision between the sampled positions in $A_i$ and
$B_j$.
\end{enumerate}

We only prove the statement i. The proof for statement ii is
similar.
%We assume that the length of motif $G$ is at least $3L$.
Note that $L$ goes down by half each cycle in the algorithm. Assume
that $L_0$ satisfies the condition of this lemma, and let $L=L_0$
happen in the algorithm.

Select $A_i$ from $S_1$ and $B_j$ from $S_2$ to be the first pair of
intervals with $||A_i(S_1,\aleph(S_1))\cap B_j(S_2,\aleph(S_2))||\ge
{L\over 2}$. It is easy to see that such a pair exists and both have
distance from the left boundary with distance at most $2L$. This is
because when an leftmost interval of length $L$ is fully inside the
motif region of the first sequence, we can always find the second
interval from the second sequence with intersection of length at
least ${L\over 2}$.

Replace $m$ by $M(L)$, $m_1$ by $M_1(L)$ (see
Definition~\ref{M-M1-def}), and $n$ by $L$ to apply
Lemma~\ref{intersection-lemma}. We do not consider any damaged
position in this algorithm, therefore, let $C$ be empty. With
probability at most $o({1\over 2^xn^3})$, there is no intersection
$A_i$ from $S_1$ and $B_j$ from $S_2$.
\end{proof}

\begin{lemma}\label{collision-all-select-lemma}
  Let $U_1$ and $U_2$ contain all positions of the input sequences
  $S_1$ and $S_2$, respectively.
  Assume $(L_{S_1}, R_{S_1}, L_{S_2}, R_{S_2})=$Collision-Detection$(S_1,U_1,S_2,,U_2)$.
   Then
\begin{enumerate}
\item
With probability at most ${1\over 2^{x}n^{3}}$, the left rough
boundary $L_{S_1}$ has at most $d_0\log n$ distance from $\LB(S_1)$
and the left rough boundary $L_{S_2}$ has at most $d_0\log n$
distance from $\LB(S_2)$.

\item
 With probability at most ${1\over 2^{x}n^{3}}$, the right rough boundary $R_{S_1}$ has at most
$d_0\log n$ distance from $\RB(S_1)$; and the right rough
%!!!(maybe should be "right rough boundary") %Fixed.
boundary of $R_{S_2}$ has at most
$d_0\log n$ distance from $\RB(S_2)$.
\end{enumerate}

\end{lemma}

\begin{proof}
For two sequences $S_1$ and $S_2$, let $\aleph(S_a)$ be the
subsequence $S_a[i_a, j_a]$ for $a=1,2$. By
Corollary~\ref{chernoff-lemma-a}, with probability at most
$P_l=2c^{d_0\log n}\le {2\over 5\cdot 2^{x}n^{3}}$ (see
inequality~\ref{param-def} at Definition~\ref{d0-sel-eqn}), there
are more than $(\alpha+\epsilon)d_0\log n$ mutations in $S_a[i_a,
i_a+d_0\log n-1]$ for $a=1,2$.

%it is easy to see that there a common position in both motif regions of
%the two sequences such that there is no mutation in the next
%$d_0\log n$ characters with high probability. This is because that
%mutation probability is small.

%By Theorem~\ref{ourchernoff-theorem}, with probability at most
%$P_{l,1}=2^{-\alpha |G|/4}$ (let $\delta=2$), there are at most
%$3\alpha
% {|G|\over 4}$ are mutated in the interval $\aleph(S_i)[1, {|G|\over 4}]$ for
%$i=1,2$.

In this case, every position in the two sequences $S_1$ and $S_2$ is
selected by Point-Selection($S_1,S_2$).
With probability at most $P_l$, %=P_{l,1}+P_{l,2}=e^{-(\log n)^{1+\Omega(1)}}$
 the left boundary position is missed during the matching.
We have similar $P_r$ %=P_{r,1}+P_{r,2}$ probability
to miss the right boundary.

Assume that $p_1$ and $p_2$ are two positions of $S_1$ and  $S_2$
respectively.  If one of two positions is outside the motif region
and has more than $d_0\log n$ distance to the motif boundary, with
probability at most $c^{-{d_0\log n}}\le {1\over 5\cdot 2^{x}n^{3}}$
(see inequality~\ref{param-def} at Definition~\ref{d0-sel-eqn}) for
them to match that requires $\diff(Y_1,Y_2)\le \beta$ by
Lemma~\ref{base-probability-lemma}, where $Y_a$ is a subsequence
$S_a[p_a,p_a+d_0\log n-1]$ for $a=1,2$.
\end{proof}

\begin{lemma}\label{collision-lemma2-2Ab} Assume that  $d_0\log n\le L\le  |G|/2$
  and $c_0\log n\le |G|<\thresholdL$.
  Let $I_1$ be a union of intervals that include $[\LB(S_1)-2L,
\LB(S_1)+2L]$ and $[\RB(S_1)-2L, \RB(S_1)+2L]$. Let
$U_1=$Point-Selection$(S_1,L, I_1)$,
$U_2=$Point-Selection$(S_2,L,[1,|S_2|])$,   and
  $(L_{S_1},R_{S_1},L_{S_2},R_{S_2})=$Collision-Detection$(S_1,U_1,S_2,,U_2)$.
   Then
\begin{enumerate}
\item
With probability at most ${1\over 2^{x}n^{3}}$, the left rough
boundary $L_{S_1}$ has at most $d_0\log n$ distance from $\LB(S_1)$
and the left rough boundary $L_{S_2}$ has at most $d_0\log n$
distance from $\LB(S_2)$.

\item
 With probability at most ${1\over 2^{x}n^{3}}$, the right rough boundary $R_{S_1}$ has at most
$d_0\log n$ distance from $\RB(S_1)$; and the right boundary of
$R_{S_2}$ has at most $d_0\log n$ distance from $\RB(S_2)$.
\end{enumerate}

\end{lemma}

\begin{proof}
In this case, every position in the two sequences $S_1$ and $S_2$ is
selected by Point-Selection($S_1,S_2$). It follows from
Lemma~\ref{collision-all-select-lemma}.
\end{proof}

\begin{lemma}\label{collision-time-lemma2}
For the case~\algtype=\randomized, we have
\begin{enumerate}
\item
 CollisionDetection($S_1, U_1, S_2, U_2$)  takes
$t(n, ||U_1||+||U_2||)=O((||U_1||+||U_2||)^2\log n)$ time.
\item
Point-Selection($S_1,L, [1,|S_1|]$) selects $s(n,L)=O(({n\over
L})M(L))$ positions in $g(n,L)=O(s(n,L))$ time if $L\ge\thresholdL$.
\item
Point-Selection($S_1,L, [1,|S_1|]$) selects $s(n,L)=O(n)$ positions
in $g(n,L)=O(n)$ time if $L<\thresholdL$.
%\item
%Initial-Rough-Boundary$(S_1, S_2, L)$ takes $O({n\over |L|}\log n)$ time.
\item
$||U_{S_{2i-1}'}||+||U_{S_j''}||$ in the algorithm \algmnam~ is no
more than $f(n,|G|)=O(M(|G|)+{n\over |G|}M(|G|))$.
\item
With probability at most ${k\over 2^xn^3}$, the algorithm \algmnam~
does not stop in $(O(k({n^2\over |G|}(\log n)^{O(1)}+h^2\log n)),
O(k))$ time.
\end{enumerate}
\end{lemma}

\begin{proof}
Statement i. The parameter $\omega_{\sublinear}$ is set to be
$\beta$ in the Collision-Detection. It follows from the time
complexity of brute force method.

Statements ii and iii. They follows from the implementation of
Point-Selection().

Statement iv.
%It follows from the brute force method time.
It follows from the choice of Point-Selection(.) for the sublinear
time algorithm at Recover-Motif(.).

 Statement iv. It follows from
Lemma~\ref{collision-all-select-lemma},
Lemma~\ref{collision-lemma2-2Ab}, Lemma~\ref{total-time-lemma}, and
Statements i, ii, and iii.
\end{proof}

%\begin{lemma}\label{InitialBoundary2}
%The second algorithm takes ${n^2\over |G|}(\log n)^{O(1)}$ time in
%Initial-Boundary($S_1,S_2$).
%\end{lemma}

%\begin{proof} Assume that $L_0$ is the largest $L$ less than
%$|G|/2$. It stops at $L_0$ with high probability.  It takes ${n\over
%L}$ intervals. For a fixed $L$. We sample $\sqrt{L}(\log n)^{O(1)}$
%points at each intervals. The time for finding the rough boundaries
%is $(({n\over L})(\sqrt{L}(\log n)^{O(1)}))^2={n^2\over |L|}(\log
%n)^{O(1)}$. The total time is $\sum_{L=L_0}^{\infty}{n^2\over
%|L|}(\log n)^{O(1)}={n^2\over |G|}(\log n)^{O(1)}$.
%\end{proof}

We give the proof for Theorem~\ref{main-theorem3}.

\begin{proof}[Theorem~\ref{main-theorem2}]
The computational time part of this theorem follows from
Lemma~\ref{collision-time-lemma2}.

By Lemma~\ref{collision-lemma-2A},
Lemma~\ref{collision-all-select-lemma}, we can let
$\varsigma_1(n)={1\over 2^xn^3}\le \varsigma_0$ for the probability
bound $\varsigma_1(n)$ in the condition~(\ref{varsigma1-condition})
of Lemma~\ref{general-lemma}.

By Lemma~\ref{collision-lemma-2A},
Lemma~\ref{collision-all-select-lemma}, we can let
$\varsigma_2(n)={1\over 2^xn^3}\le \varsigma_0$ for the probability
bound $\varsigma_2(n)$ in the condition~(\ref{varsigma1-condition})
of Lemma~\ref{general-lemma}.

By inequality~(\ref{support-P0-Q0-inequality}),
%(\ref{v-alpha0-ineqn}),
the condition~(\ref{P0-Q0-condition}) of Lemma~\ref{general-lemma}
is satisfied.

 By inequality~(\ref{v-set-ineqn1}), we know that the
condition~(\ref{V_0-condition}) of Lemma~\ref{general-lemma} can be
satisfied.

% and Lemma~\ref{total-time-lemma}.
 The failure
probability part of this theorem follows from
%Lemma~\ref{collision-lemma}, Lemma~\ref{collision-lemma2},
Lemma~\ref{v+u-lemma2}, and Lemma~\ref{general-lemma} by using the
fact that $k_1,k_2$, and $k$ are of the same order (see
equation~(\ref{Z1-Z2-eqn})).
\end{proof}

\subsection{Deterministic  Algorithm for $\Omega(1)$ Mutation
Rate}\label{deterministic-sec}

 In this section, we give a deterministic algorithm for
the case with $\Omega(1)$ mutation rate.  The performance of the
algorithm is stated in Theorem~\ref{main-theorem3}.

%The parameter $\omega$ is set to be $\beta$ in the
%Collision-Detection. So, the Collision-Detection$(S_1, U_1, S_2,
%U_2)$ here is the same as that as Collision-Detection$(S_1, U_1,
%S_2, U_2)$ for the second algorithm.

%\subsection{Analysis of the Third Algorithm}

\begin{lemma}\label{collision-lemma2-3A} Assume that  $d_0\log n\le L\le  |G|/2$
  and $c_0\log n\le |G|$. Let $I_1$ be a union of intervals that include $[\LB(S_1)-2L,
\LB(S_1)+2L]$ and $[\RB(S_1)-2L, \RB(S_1)+2L]$. Let
$U_1=$Point-Selection$(S_1,L, I_1)$,
$U_2=$Point-Selection$(S_2,L,[1,|S_2|])$,   and
  $(L_{S_1},R_{S_1},L_{S_2},R_{S_2})=$Collision-Detection$(S_1,U_1,S_2,,U_2)$.
  Then
\begin{enumerate}
\item
With probability at most ${1\over 2^{x}n^{3}}$, the left rough
boundary $L_{S_1}$ has at most $d_0\log n$ distance from $\LB(S_1)$
and the left rough boundary $L_{S_2}$ has at most $d_0\log n$
distance from $\LB(S_2)$.

\item
 With probability at most ${1\over 2^{x}n^{3}}$, the right rough boundary $R_{S_1}$ has at most
$d_0\log n$ distance from $\RB(S_1)$; and the right rough
%!!!(should be "right rough boundary")  %Fixed.
boundary of $R_{S_2}$ has at most
$d_0\log n$ distance from $\RB(S_2)$.
\end{enumerate}

\end{lemma}

\begin{proof}
In this case, every position in the two sequences $S_1$ and $S_2$ is
selected by Point-Selection($S_1,S_2$). It follows from
Lemma~\ref{collision-all-select-lemma}.
\end{proof}

\begin{lemma}\label{collision-time-lemma3}
For the case~\algtype=\deterministic, we have
\begin{enumerate}
\item
 CollisionDetection($S_1, U_1, S_2, U_2$)  takes
$t(n, ||U_1||+||U_2||)=O((||U_1||+||U_2||)^2\log n)$ time.
\item
Point-Selection($S_1,L, [1,|S_1|]$) selects $s(n,L)=O(n)$ positions
in $g(n,L)=O(n)$ time.

%\item
%Initial-Rough-Boundary$(S_1, S_2, L)$ takes $O(n\log n)$ time.
\item
$||U_{S_{2i-1}'}||+||U_{S_j''}||$ in the algorithm \algmnam~ is no
more than $f(n,|G|)=O(|G|+n)$.

\item
With probability at most ${k\over 2^xn^3}$, the algorithm \algmnam~
does not stop $(O(k(n^2(\log n)^{O(1)}+h^2\log n)), O(k))$.
\end{enumerate}
\end{lemma}

\begin{proof}
Statement i. The parameter $\omega_{\deterministic}$ is set to be
$\beta$ in the Collision-Detection. It follows from the time
complexity of brute force method.

Statement ii. They follows from the implementation of
Point-Selection().

Statement iii.
%It follows from the brute force time.
It follows from the choice of Point-Selection(.) for the sublinear
time algorithm at Recover-Motif(.).

Statement iv. It follows from Lemma~\ref{collision-lemma2-3A},
Lemma~\ref{total-time-lemma} and Statements i, ii, and iii.
\end{proof}

%\begin{lemma}\label{Initial-Boundary3}
%The third algorithm takes ${n^2}(\log n)^{O(1)}$ time in
%Initial-Boundary($S_1,S_2$).
%\end{lemma}

%\begin{proof}
%It has $O(n^2)$ possible cases to find left and right collisions.
%Each comparison takes $O(\log n)$ time.
%  The total time is ${n^2}(\log n)^{O(1)}$.
%\end{proof}

We give the proof for Theorem~\ref{main-theorem3}.

\begin{proof}[Theorem~\ref{main-theorem3}]
The computational time part of this theorem follows from
Lemma~\ref{collision-time-lemma3}.
% and Lemma~\ref{total-time-lemma}.

By Lemma~\ref{collision-lemma2-3A}, we let $\varsigma_1(n)={1\over
2^xn^3}\le \varsigma_0$ for the probability bound $\varsigma_1(n)$
in the condition~(\ref{varsigma1-condition}) of
Lemma~\ref{general-lemma}.

By Lemma~\ref{collision-lemma2-3A}, we can let
$\varsigma_2(n)={1\over 2^xn^3}\le \varsigma_0$ for the probability
bound $\varsigma_2(n)$ in the condition~(\ref{varsigma1-condition})
of Lemma~\ref{general-lemma}.

By inequality~(\ref{support-P0-Q0-inequality}),
%(\ref{v-alpha0-ineqn}),
the condition~(\ref{P0-Q0-condition}) of Lemma~\ref{general-lemma}
is satisfied.

By inequality~(\ref{v-set-ineqn1}), we know that the
condition~(\ref{V_0-condition}) of Lemma~\ref{general-lemma} can be
satisfied.

% and Lemma~\ref{total-time-lemma}.
 The failure
probability part of this theorem follows from
%Lemma~\ref{collision-lemma}, Lemma~\ref{collision-lemma2},
Lemma~\ref{v+u-lemma2}, and Lemma~\ref{general-lemma} by using the
fact that $k_1,k_2$, and $k$ are of the same order (see
equation~(\ref{Z1-Z2-eqn})).
%It follows from Lemma~\ref{general-lemma}, Lemma~\ref{total-time-lemma}, and Lemma~\ref{InitialBoundary3}.
\end{proof}

\section{Conclusions}
We develop an algorithm that under the probabilistic model. It finds
the implanted motif with high probability if the alphabet size is at
least $4$, the motif length is in $[(\log n)^{7+\mu}, {n\over (\log
n)^{1+\mu}}]$ and each character in motif region has probability at
most ${1\over (\log n)^{2+\mu}}$ of mutation. The motif region can
be detected and each motif character can be recovered in sublinear
time. A sub-quadratic randomized algorithm is developed to recover
the motif with $\Omega(1)$ mutation rate. A quadratic deterministic
algorithm is developed to recover the motif with $\Omega(1)$
mutation rate. It is interesting problem if there is an algorithm to
handle the case for the alphabet of size $3$. A more interesting
problem is to extend the algorithm to handle larger mutation
probability.

 \section*{Acknowledgements} We thank Ming-Yang Kao for introducing
 us to this topic. We also thank
 Lusheng Wang and Xiaowen Liu for some discussions. We would like to
 thank Eugenio De Hayos to his helpful comments.
 %We are also grateful to
%unknown referee's who helped us improve the presentation of this paper.

\section{Experimental Results}

\subsection{Implementation and Results}
\noindent Aiming at solving the motif discovery problem, we
implemented our algorithm by JAVA. Our program could accept many
popular DNA sequence data formats, such as FASTA,GCG, GenBank and so
on. Our tests were all done on a PC with an Intel Core 1.5G CPU and
3.0G Memory.

In the first experiment, we tested our  algorithm on several sets of
simulated data, which are all generated from our probability model
with a small mutation rate. All input sets contain 20 or 15
sequences, each of length is 600 or 500 base pair. And each bp of
all the simulated gene sequences was generated independently with
the same occurrence probability. A motif with a length of 15 or 12
was randomly planted to each input sequence. The number of
iterations is between 10 and 30. The minimum Hamming distances
between the results and consensus are recorded.

\begin{center}
\begin{tabular}{|c|c|c|c|c|c|c|}
\hline
 &N & M & L & R & Accuracy & timecost \\ [0.5ex]
\hline
Sets1&20 & 600 & 15 & 10 & 100 & 23\\
%\hline
Sets1&20 & 600 & 15 & 30 & 100 & 85\\
%\hline
Sets2&15 & 600 & 15 & 10 & 95 & 18\\
%\hline
Sets3&20 & 600 & 12 & 10 & 95 & 15\\
%\hline
Sets4&20 & 500 & 15 & 20 & 100 & 112\\
\hline

\end{tabular}
\end{center}

\begin{center}
Tab 1. Results on simulated data
\end{center}
\

In the second experiment,we tested our algorithm on real sets of
sequences, which are obtained from SCPD. SCPD contains a large
number of gene data and transcription factors of yeast. For each set
of gene sequences that are regulated by the same motif, we chose
1000bp as the length of input gene sequence. In order to make
comparisons among several existed motif finding methods, we also
tested Gibbs, MEME, Info-Gibbs and Consensus in our experiment to
show the difference of their performance. Here are the specific
experimental results:

\begin{center}
\begin{tabular}{|c|c|c|}
\hline
 & Number of Sequences & Motif Length\\ [0.5ex]
\hline
GCR1 & 6 & 10\\
\hline
\end{tabular}
\end{center}

\begin{center}
Tab 2. Number of sequences and motif length
\end{center}

%\begin{tabular}{|c|c|c|c|c|c|c|}
%\hline

\
\begin{center}
\begin{tabular}{|c|c|}
\hline & GCR1\\ [0.5ex] \hline
Voting& 4 \\
Gibbs& 5 \\
MEME& 10\\
InfoGibbs& 5\\
Consensus& 5\\
\hline
\end{tabular}
\end{center}
\begin{center}
Tab 3. Total number of mismatch positions compared to motif
\end{center}

\
\begin{center}
\begin{tabular}{|c|c|}
\hline &GCR1\\ [0.5ex] \hline
Voting& 0.67\\
Gibbs& 0.83\\
MEME& 1.67\\
InfoGibbs& 0.83\\
Consensus& 0.83\\
\hline
\end{tabular}
\end{center}

\begin{center}
Tab 4. Average mismatch numbers per sequence
\end{center}

\subsection{Analysis}
In the first experiment, we tested our algorithm on 4 sets of
simulated sequences. From our experimental results, we find that the
accuracy of our algorithm for finding motif in simulated data is
nearly 100\%. The accuracies of the experiments in simulated data
sets are satisfactory. Our algorithm can get the results within
several minutes.

From the second experimental results, we find that our  algorithm is
able to find the real motifs from given gene sequences in little
time. Our algorithm shows higher speed than other four motif finding
methods, because an initial motif pattern is first extracted from
comparing two sequences in the first stage of our algorithm. In
addition, unlike Gibbs sampling and EM methods, our algorithm could
avoid some extra time consuming computations, such as the
calculations of likelihoods. According to this feature, we use the
consensus string of the voting operation obtained from the result of
last iteration as a new starting pattern to program, and continue
doing voting until there is no further improvement. Experimental
results show that if we set the number of iterations to be large
enough, the program could give more accurate results in reasonable
time. Besides, in order to detect unknown motifs in sequences, our
program also provides several possible motifs existed in specific
sequences, and the average mismatch numbers of motifs that is
greatly lower than other four methods.

\section{Future works}
Compared with other tested motif finding methods, we could find that
the voting algorithm has advantages in some aspects, but there are
still some improvements could be done on this algorithm. As we know,
though a set of sequences may have the consensus, but each motif in
sequence may has mutations, and the length of each motif could also
be different. So the two factors increase the difficulties in
finding unknown motifs. In the future, we plan to improve the
efficiency of voting algorithm by studying other motif finding
methods, such as MEME, a combination may be made between voting
algorithm and MEME so that voting algorithm could have better
performance in finding unknown motifs.\\

%\bibliographystyle{abbrv}
%\bibliography{bib}

%\newpage

%{\LARGE Appdenix}

%\section{Appendix: \algmb for the Small Alphabet}

%\appendix{Proof for Lemma~\ref{phase2b-lemma}}
%\section{Appendix 1: Proof of Lemma~\ref{phase2b-lemma}}

%\newpage

%\bibliographystyle{abbrv}
%\bibliography{bib}

\end{document}